\documentclass{cta-author}
\usepackage{amsmath}
\usepackage{verbatim}
\usepackage{ragged2e}
\usepackage{pifont}
\usepackage{subcaption}
\usepackage{caption,psfrag}
\usepackage{amssymb}
\usepackage{amsthm}
\usepackage{float}
\usepackage{mathtools}
\usepackage{mathtools}
\usepackage[english]{babel}
\usepackage[english]{babel}
\newtheorem{theorem}{Theorem}[section]

\newtheorem{lemma}[theorem]{Lemma}
\newtheorem{remark}{Remark}
\usepackage{amsthm}
\theoremstyle{plain}
\newtheorem{assumption}{Assumption}

\begin{document}
%
\title{Variable Gain Gradient Descent-based Reinforcement Learning for Robust Optimal Tracking Control of Uncertain Nonlinear System with Input-Constraints}
%
%
%

\author{\au{Amardeep Mishra$^{1\corr
}$}, \au{Satadal Ghosh$^{2}$}}

\address{\add{1}{Student, Department of Aerospace Engineering, IIT Madras, Chennai 600036, India, ae15d405@smail.iitm.ac.in}
\add{2}{Faculty, Department of Aerospace Engineering, IIT Madras, Chennai 600036, India, satadal@iitm.ac.in}
\email{ae15d405@smail.iitm.ac.in}}

\begin{abstract}
In recent times, a variety of Reinforcement Learning (RL) algorithms have been proposed for optimal tracking problem of continuous time nonlinear systems with input constraints. 
Most of these algorithms are based on the notion of uniform ultimate boundedness (UUB) stability, in which normally higher learning rates are avoided in order to restrict oscillations in state error to smaller values. However, this comes at the cost of higher convergence time of critic neural network weights. 
This paper addresses that problem by proposing a novel tuning law containing a variable gain gradient descent for critic neural network that can adjust the learning rate based on Hamilton-Jacobi-Bellman (HJB) approximation error. 
By allowing high learning rate the proposed variable gain gradient descent tuning law could improve the convergence time of critic neural network weights. Simultaneously, it also results in tighter residual set, on which trajectories of augmented system converge to, leading to smaller oscillations in state error. 
A tighter bound for UUB stability of the proposed update mechanism is proved.
Numerical studies are then presented to validate the robust Reinforcement Learning control scheme in controlling a continuous time nonlinear system. 
\end{abstract}
\maketitle

%

\section{Introduction}
%
%
%
%
Optimal Control as a part of Control Theory seeks to minimize the cost function subjected to system dynamics as constraints. 
In nature, oragnisms act on the environment and observe the resulting reward. 
Over time, the actions on the environment are tweaked to improve the likelihood of the reward. 
This process is sometimes referred to as reinforcement learning and it captures the notion of optimality \cite{lewis2009reinforcement}.
Adaptive dynamic programming (ADP) refers to the mathematical formulation and solution of the reinforcement learning problem \cite{lewis2009reinforcement}. 
It is a practical way of implementing optimal control.
The optimal control problem can be broadly classified into two major categories: regulation Problems (wherein states are driven to zero) and Trajectory Tracking Problems (wherein error between actual state and desired state is driven to zero). For a general nonlinear system, optimal control requires the solution of Hamilton Jacobi Bellman (HJB) equation (which is a nonlinear partial differential equation (PDE)) that yields the optimal cost function. The optimal value function is then used to generate optimal control action. The fundamental problem with this approach is that in even simplest of nonlinear cases, the HJB equation is extremely difficult to solve. For linear systems though, HJB equation is transformed into Riccati equation. 

In order to alleviate the challenge of solving HJB equation directly, iterative Approximate Dynamic Programming methods were first proposed in the works of Werbos as a method to solve optimal control problem for discrete time (DT) systems in his seminal work \cite{werbos1974beyond}- \cite{werbos1977advanced}. 
Neural Network (NN) were used to deal with unknown functions. 
Sutton and Barto in 1995 proposed ADP for discrete time systems \cite{barto19951}. 
The usage of two distinct NNs (also known as Actor-Critic structure) to learn the cost function and the control action first appeared in the works of Barto \cite{barto1983neuronlike} where both the NNs were tuned online. 
Werbos came up with a third NN to approximate the system dynamics \cite{werbos1989neural}. 
All of the aforementioned works deal with DT systems and the first few works that appeared for generic nonlinear continuous time system are from  \cite{abu2005nearly}, \cite{lin2011optimality},  \cite{liu2013adaptive}, \cite{vamvoudakis2010online}, \cite{murray2002adaptive} \cite{yang2014online} \cite{zhao2014mec} \cite{modares2014integral} \cite{bhasin2013novel}. 
The works mentioned above deal with continuous time (CT) nonlinear optimal regulation problem where states are driven to zero. 
Vamvoudakis and Lewis (2010) deal with CT nonlinear optimal regulation problem based on an online algorithm, which involved tuning of the critic and the actor weights in a synchronous fashion. 
The apriori knowledge of the CT nonlinear system dynamics is assumed in both Abu-Khalaf and Lewis (2005) and Vamvoudakis and Lewis (2010). Bhasin et al \cite{bhasin2013novel} introduced a novel method of computing control action for regulation problems for CT nonlinear system where partial knowledge of system dynamics exists. Their method demanded the knowledge of control gain matrix. 
The primary advantage of their methodology was simultaneous tuning of the actor and the critic. Nonetheless, a predefined convex set was required in their work for the implementation of the projection algorithm. This was done to force the NN weights to remain in the set. 
Identifier NNs were used in the works of Yang et al. \cite{yang2014reinforcement} to obviate the requirement of knowledge of drift dynamics. This technique could generate the optimal control for nonlinear continuous time systems with unknown structures.

Use of identifiers is not the only method that has been proposed in the literature to deal with unknown systems while implementing ADP. 
Integral Reinforcement Learning (IRL), first proposed by Vrabie et al. \cite{vrabie2009adaptive} is one such implementation of RL wherein the system dynamics knowledge is not required in policy evaluation step, i.e., the step involving the evaluation of cost function. 
However, it too requires the knowledge of control dynamics in policy iteration step, i.e, the step involving generation of control action. 
Synchronous tuning of actor-critic NN, based on a novel IRL algorithm was first proposed by Modares et al. \cite{modares2014integral} in 2014 for continuous time nonlinear systems. 
A robust ADP algorithm was proposed by Jiang and Jiang \cite{jiang2014robust} to derive the robust control for uncertain nonlinear systems. It was achieved by synthesizing the optimal control solution with infinite horizon cost for original uncertain nonlinear system. 
However, like most of the ADP methods introduced above, Jiang's formulation \cite{jiang2014robust} required initial stabilizing controller. 

Most of the aforementioned ADP schemes are for regulation problems. 
ADP formulations for optimal tracking control problem (OTCP) for CT nonlinear systems was initially proposed by Zhang et al. \cite{zhang2011data} in 2011. 
Zhang's method entailed two different controllers viz., the adaptive optimal control (for transient behaviour, i.e to stabilize the tracking error in transience in an optimal manner) and steady state controller (for steady state, i.e, to maintain the tracking error close to zero in steady state). 
However, the major limitation of his method was that it required the control gain matrix to be invertible in order to implement a steady state controller. 
This requirement was relaxed in \cite{heydari2014fixed} in 2014 when they proposed a single network based critic structure to approximate the cost function. 
Thereafter, \cite{modares2014optimal} proposed an algorithm that was used to analyze the constrained-input optimal tracking problem with a discounted value function for CT nonlinear systems. 
It should be mentioned that the knowledge of drift dynamics is not required in  \cite{modares2014integral} and  \cite{modares2014optimal}, however the knowledge of control dynamics is assumed). Most of the schemes discussed above require an initial stabilizing control to initiate the process of policy iteration.

Finding an initial stabilizing controller to begin the policy improvement is often a very difficult task. 
Recently, a way to relax the criteria of initial stabilizing control for ADP based RL methods (policy iteration) was proposed by Dierks and Jagannathan \cite{dierks2010optimal} as a single online approximator based system. 
Similarly, Yang el al.  \cite{yang2015robust} proposed an ADP algorithm for robust optimal tracking control of nonlinear systems in 2015. 
This formulation, did not require an initial stabilizing controller for robust optimal tracking control problem for nonlinear systems.  
In order to approximate the value function, a single critic NN was utilized in their paper. 
Tracking control action was generated by critic NN. 
However, their method requires the knowledge of nominal plant dynamics and does not include the input constraints. 
It is also noted that their method took a lot of time to achieve convergence of critic NN weights and reduction of oscillation magnitude in state error to a small residual set. 
These requirements might not be feasible for a lot of practical cases. 

\par
Inspired by \cite{yang2015robust} and \cite{liu2015reinforcement}, this paper addresses these concerns by proposing a ADP-based robust optimal tracking controller that is driven by a novel variable gain gradient descent tuning law. 
Similar to \cite{yang2015robust} and \cite{liu2015reinforcement}, the critic update law is made up of three terms, the first term is responsible for reducing the HJB error, the second term is responsible for stability, i.e., it comes into effect when the Lyapunov function is growing along the augmented system trajectories and lastly the third term determines the size of the compact UUB set on which the augmented states finally converge to. 
However, unlike \cite{yang2015robust} and \cite{liu2015reinforcement}, the learning rate of gradient descent presented in this paper is a function of HJB error.
This leads to improved tracking performance in terms of faster convergence times of critic neural network weights and smaller oscillation magnitude of state error (error between actual state and desired state).

The salient features of the proposed variable gain gradient descent scheme for RL tracking controller are:\\
(i) To the best of authors' knowledge, this is the first time when variable learning rate is leveraged in gradient descent to tune critic NN weights to solve robust optimal tracking problem for continuous time nonlinear systems with actuator constraints.
The first term in the weight update law responsible for reducing the approximate HJB error is driven by variable learning rate gradient descent where, the variable learning rate is a function of HJB error. 
So, when HJB error is large, the learning rate gets scaled up proportionally which results in speedier reduction in HJB error, however the learning process is dampened, as the HJB error approaches zero.\\
(ii) Further, variable gain gradient descent leads to tighter residual set for critic NN weights thus resulting in approximated optimal controllers that are closer to ideal optimal controller.
This in turn leads to improved tracking performance.

The rest of the paper is organized as follows, Section \ref{sec2} introduces robust optimal tracking controller and its preliminaries. 
This section is followed by Section \ref{sec3} that utilizes the concept of RL to solve optimal trajectory tracking problem for continuous time nonlinear system with actuator constraints. 
It is divided into two subsections (subsection \ref{crit_approx} and \ref{param_update}) that delve into value function approximation using critic NN and existing parameter update law respectively.
It is then followed by Section \ref{sec:var_gain} that presents the primary contribution of this paper i.e., novel weight update law for critic NN.
This section contains subsection \ref{stab} that discusses the stability proof of the update law presented in this paper in detail.
Finally towards the end, the paper is concluded by \ref{result} and \ref{conclusion} that discusses results and conclusions respectively.

\section{Robust Optimal Tracking Controller}\label{sec2}
\subsection{Problem Formulation}\label{prelim}


The uncertain nonlinear dynamics is given by the affine-in-control equation,
\begin{equation}
\dot{x}=f(x)+g(x)u+\Delta f(x)
\label{augm}
\end{equation}
where $f(x)$ and $g(x)$ are known dynamics (drift and control coupling dynamics) and $\Delta f(x)$ is the unknown matching perturbation.

\vspace{0.2cm}
\begin{assumption}\label{uncert}
\textnormal{The drift dynamics i.e., $f(x)$ is Lipschitz continuous in $x\in \Omega \subset \mathbb{R}^n$ and $g(x)$ is bounded such that, $\exists g_M>0 \ni 0<\|g(x)\|<g_M,~ \forall x \in \mathbb{R}^n$. 
It is also assumed that matching condition is satsified by the perturbation i.e., $\Delta f(x)=g(x)d(x)$, where $d(x) \in \mathbb{R}^m$ is an unknown function bounded by a known function $d_M (x) > 0$.}
\end{assumption}

\begin{assumption}\label{des_dyn}
\textnormal{The commanded trajectory, i.e, the $\dot{x}_d(t):~\mathbb{R} \to \mathbb{R}^n$ is bounded and is Lipschitz continuous satisfying $H(0)=0$, $\ni$, $\dot{x}_d=H(x_d)$.}
\end{assumption}
These two assumptions are in line with Assumption 2 of \cite{yang2015robust} and  Assumptions 1, 2 and 3 of \cite{liu2015reinforcement}.

\textbf{Objective of the control}: It is required to derive a robust optimal tracking controller that makes the system trajectories $x$ follow the desired reference trajectory $x_d$ with state error in a sufficiently small neighborhood of the origin in the presence of unknown but bounded $d(x)$.

\subsection{Preliminaries of Robust OTCP}
In \cite{liu2015reinforcement} the feedback controller was derived for constrained input case in the presence of unknown uncertainties for optimal regulation problem.
In this paper optimal tracking problem is considered with actuator constraints and unknown uncertainties.
In order to achieve the desired objective, an augmented system dynamics that consists of dynamics of errors ($\dot{e}$) and desired states ($\dot{x}_d$) is defined first. 
Using (\ref{augm}) and Assumptions (\ref{uncert}) and (\ref{des_dyn}), tracking error dynamics can be written as:
\begin{equation}
\begin{split}
\dot{e}&=\dot{x}-\dot{x}_d \\
\dot{e}&=f(x_d+e)+g(x_d+e)u(t)-H(x_d(t))+\Delta f(x_d+e)
\end{split}
\label{e_dyn}
\end{equation}
Therefore, the dynamics of augmented system, given as $z=[e^T,x_d^T]^T$, can compactly be written as:
\begin{equation}
\dot{z}={F}(z)+{G}(z)u+\Delta F(z)
\label{aug}
\end{equation}
where, $u \in \mathbb{R}^m$, ${F}:\mathbb{R}^{2n} \to \mathbb{R}^{2n}$ and ${G}: \mathbb{R}^{2n} \to \mathbb{R}^{2n\times m}$ are given by:
\begin{equation}
\begin{split}
{F}(z)=\begin{pmatrix}
{f}(e+x_d)-H(x_d) \\
H(x_d)
\end{pmatrix},~~ {G}(z)=
\begin{pmatrix}
{g}(e+x_d) \\
0
\end{pmatrix}
\end{split}
\label{FG}
\end{equation}
 $\Delta F(z) \in \mathbb{R}^{2n}$ and is defined as $\Delta F(z)=G(z)d(z)$ with $d(z) \in \mathbb{R}^m$ and $\|d(z)\| \leq d_M(z)$. 
Following Assumptions \ref{uncert} and \ref{des_dyn} and Eq. \eqref{FG}, $\|F(z)\| \leq L_f\|z\|$ and $\|G(z)\| \leq g_M$.
 In the subsequent analysis, $d_M \triangleq d_M(z)$.

One of the prime advantages of creating an augmented system, is that, the controller does not require invertibility of control gain matrix and a single controller comprising of both steady state controller and transient control can be synthesized \cite{modares2014optimal} \cite{kiumarsi2014reinforcement}.
Nominal augmented dynamics is given by:
\begin{equation}
\dot{z}={F}(z)+{G}(z)u
\label{nad}
\end{equation}

The infinite horizon discounted cost function for (\ref{nad}) is considered as follows \cite{modares2014optimal} :
\begin{equation}
V(z)=\int_t^{\infty} e^{-\gamma(\tau-t)}[ d_M^2+\bar{u}(z,u)]d\tau    
\label{cost}
\end{equation}
where, $\bar{u}=z^TQ_1z+C(u)$ is the utility function comprising of augmented state $z$ and the control action $u$. 
The positive definite diagonal matrix $Q_1 \in \mathbb{R}^{2n\times 2n}$ is defined as,
\begin{equation}
\begin{split}
Q_1=\begin{pmatrix}
Q & 0_{n\times n} \\
0_{n\times n} &0_{n\times n}
\end{pmatrix}
\end{split}
\end{equation}
where, $Q \in \mathbb{R}^{n\times n}$ is a positive definite diagonal matrix with non zero entries. 
\begin{remark}
\textnormal{Following from \cite{yang2015robust} and \cite{liu2015reinforcement}, it can be shown that robust control problem for \eqref{aug} can be transformed into optimal tracking control problem for nominal augmented system \eqref{nad} with discounted cost function \eqref{cost}.}
\end{remark}
In trajectory tracking problems, $x_d$ contained in $z$ might not go to $0$ in steady state and $u$ encapsulates both optimal part and steady state part, hence, infinite horizon cost index comprising of $z,u$ might blow up and become infinite.
Hence, in order to make $V$ finite, discounted cost function of the form (\ref{cost}) is chosen for trajectory tracking problems.
Generally, the function $C(u)$ is quadratic in nature, however, it can be non-quadratic \cite{abu2008neurodynamic}, \cite{modares2013adaptive}, if, control constraints are taken into account, i.e, $|u_i| \leq u_m,~i=1,2...m$. This corresponds to an input-constrained scenario, which is also considered in this paper. 
Thus, $C(u)$ is defined in this paper as follows \cite{abu2005nearly},\cite{abu2008neurodynamic}-\cite{lyashevskiy1996constrained}.

\begin{equation}
\begin{split}
C(u)&=2u_m\int_0^u (\psi^{-1}(\nu/u_m))^TRd\nu \\
&=2u_m\sum_{i=1}^{m}\int_0^{u_i} (\psi^{-1}(\nu_i/u_m))^TR_i d\nu_i
\end{split}
\label{ctr_cu}
\end{equation}
where, $R \in \mathbb{R}^{m\times m}$ is a positive definite matrix, $\psi \in \mathbb{R}^m$ is a function possessing following properties \\
(i) It is odd and monotonically increasing \\
(ii) It is bounded function ($|\psi(.)| \leq 1$) that belongs to $C^p(p \geq 1)$. 
In literature dealing with constrained input, some of the possible candidates for $\psi$ include, $tanh,erf,sigmoid$. 
In this paper $\psi^{-1}(.)=\tanh^{-1}(.)$.
It can be clearly observed that $C(u)$ (as shown in Lemma \ref{cu} in appendix) is positive. 
The discount factor, $0 \leq \gamma$, defines the value of utility in future. 
The first term inside the integral caters to any perturbations or uncertainties that might appear in the plant dynamics. 

Differentiating (\ref{cost}) along the nominal system trajectories the following can be obtained \cite{modares2014integral}:
\begin{equation}
\begin{split}
 V_z(z)({F}(z)+{G}(z)u)-\gamma V(z)+ d_M^2+\bar{u}(z,u)\\
=\mathcal{H}(z,u,V_z(z))=0   
\end{split}
\label{Ha}
\end{equation}
where, $\mathcal{H}(.)$ represents the Hamiltonian and $V_z(z) \triangleq \nabla_z V(z)$.
Let $V^*(z) \in C^1$ be the optimal cost function that satisfies $\mathcal{H}(.)=0$ and is given by:
\begin{equation}
V^*(z)=\min_{u} \int_t^{\infty}e^{-\gamma(\tau-t)}[d_M^2+\bar{u}(z,u)]d\tau
\label{opt_cost}
\end{equation}
Also in the subsequent analysis, $V\triangleq V(z)$, $V^* \triangleq V^*(z)$ and $V^*_z \triangleq V^*_z(z)$.
Thus, $\mathcal{H}(.)=0$ can be re-written in terms of optimal cost as:
\begin{equation}
\begin{split}
 \nabla_z{V^*}({F}(z)+{G}(z)u)-\gamma V^*(z)+d_M^2+\bar{u}(z,u)=0
\end{split}
\label{hjb1} 
 \end{equation}

Differentiating (\ref{hjb1}) with respect to (w.r.t.) $u$, i.e, $\partial \mathcal{H}/ \partial u=0$, closed form of optimal control action $u^*$ is obtained as:
\begin{equation}
u^*=-u_m\tanh{\Big(\frac{1}{2u_m}R^{-1}{G}(z)^T\nabla_z{V^*}\Big)}
\label{opt_u}
\end{equation}

Substituting (\ref{opt_u}) in (\ref{hjb1}) the HJB equation is formulated as:
\begin{equation}
\begin{split}
V^*_z{F}(z)-2u_m^2A^T(z)\tanh(A(z))+d_M^2+z^TQ_1z+\\
2u_m\int_0^{u^*}\tanh^{-1}(\nu/u_m)^TRd\nu-\gamma V^*=0
\end{split}
\label{hjb_int}
\end{equation}
where $V_z^*=\nabla_z{V^*}$, and $A=(1/2u_m)R^{-1}{G}(z)^TV^*_z \in \mathbb{R}^m$.
The $C(u)$ or last but one term in left hand side of (\ref{hjb_int}) can be simplified as:
\begin{equation}
\begin{split}
&2u_m\int_0^{-u_m\tanh{A(z)}}\tanh^{-1}(\nu/u_m)^TRd\nu\\
&=2u_m^2A^T(z)R\tanh{A(z)}+
u_m^2\sum_{i=1}^{m}R_i\ln[1-\tanh^2{A_i(z)}]
\end{split}
\label{hjb_int1}
\end{equation}
Eq. (\ref{hjb_int1}) follows from Lemma \ref{cu1} given in Appendix \ref{L1}.  Now, using (\ref{hjb_int1}), Eq. (\ref{hjb_int}) can further be simplified into:
\begin{equation}
\small
 V^*_z{F}(z)+d_M^2+z^TQ_1z+u_m^2\sum_{i=1}^{m}R_i\ln[1-\tanh^2{A_i(z)}]-\gamma V^*=0   
\label{hjb_sum}
\end{equation}
Eq. (\ref{hjb_sum}) is the HJB equation which is a nonlinear PDE in optimal cost function. Note that $\ln{(.)}$ used in this paper is natural log with base $e$. Now, the optimality of $u^*$ (defined in (\ref{opt_u})) and asymptotic stability of ($e=x-x_d$) would be discussed, which follows the line of logic as Theorem 1 of \cite{modares2014optimal}.

\begin{theorem}
For augmented system defined in (\ref{aug}), and its associated discounted cost function defined in (\ref{cost}) with $V^*$ being the solution of the HJB equation, the controller ($u^*$) described as in (\ref{opt_u}) minimizes the performance index (\ref{cost}) over all control policies constrained to $|u_i| \leq u_m$.
Further, it also ensures asymptotic stability of error dynamics (\ref{e_dyn}) in the limiting sense when $\gamma \to 0$.
\end{theorem}

\begin{proof}
In the formulation of robust OTCP, it can be observed from (\ref{cost}) and (\ref{Ha}) that both $V(z)$ and $\mathcal{H}$ contain $d_m$ (upper bound of unknown perturbation $d(z)$) in their expressions, respectively. This is the difference w.r.t. the expressions of $V(z)$ and $\mathcal{H}$ corresponding to OTCP in \cite{modares2014optimal}. It is evident that the presence of the robust term $d_M^2$ in the performance index does not alter the proof for optimality of $u^*$. Therefore, for details of this part of the proof refer to the proof of Theorem 1 in \cite{modares2014optimal}.

However, the presence of the robust term $d_M^2$ in the performance index affect the proof of the stability in the following way.
Differentiating $V^*(z)$ along the augmented system trajectories,
\begin{equation}
\begin{split}
 \nabla_z{V^*}({F}(z)+{G}(z)u+\Delta F(z))-\gamma V^*(z)+d_M^2+\\
 \bar{u}(z,u)=0
\end{split}
\label{hjb2} 
 \end{equation}
Multiplying $e^{-\gamma t}$ to both sides of Eq. (\ref{hjb2}),
\begin{equation}
\begin{split}
\frac{d e^{-\gamma t}V(z(t))}{dt}=-e^{-\gamma t}[d_M^2+\bar{u}(z,u^*)] \leq 0
\end{split}
\label{dm}
\end{equation}
Note that, (\ref{dm}) has an additional term ($d_M^2$) on the right hand side (RHS) of the equation compared to the Eq. (43) of \cite{modares2014optimal}.
Therefore, following the same methodology as \cite{modares2014optimal}, it can be observed that, tracking error is asymptotically stable when $\gamma= 0$.
However, when $\gamma \neq 0$, the considering $V^*$ and $C(u^*)$ to be finite and using the fact that $z^T Q_1z=e^TQe$, the stability can be analyzed in two cases, namely,

Case (a): When $A<0 \Rightarrow \gamma<\frac{d_M^2+C(u^*)}{V^*}$, then in this case, $\dot{V}<0$, for all values of $e$ implying asymptotic stability.
In order to ensure asymptotic stability, use of sufficiently small value of $\gamma$ is suggested.

Case (b): When $A>0\Rightarrow \gamma>\frac{d_M^2+C(u^*)}{V^*}$, then $\dot{V}<0$, only when following inequality is satisfied, i.e.,
\begin{equation}
\|e\| > \sqrt{\frac{\gamma V^*-d_M^2-C(u^*)}{\lambda_{min}(Q)}}
\label{uub_E}
\end{equation}
The RHS of Eq. (\ref{uub_E}) gives the UUB bound for state error $e$, which is valid only when $A>0$.
\end{proof}

\vspace{0.2 cm}
\section{Background of Optimal Tracking Using RL}\label{sec3}
\subsection{Approximation of Value function using Critic NN}\label{crit_approx}

For applying the optimal controller (\ref{opt_u}), $V^*$ must be calculated first. 
This is difficult to achieve because it requires solution to (\ref{hjb1}), which is a nonlinear PDE. In order to by-pass solving the HJB equation directly, an NN will be utilized to approximate the value function. For that, in this paper, the value function is assumed to be smooth. Let there exist ideal weight parameter vector $W$ that can accurately approximate the value function as:
\begin{equation}
V^*(z)=W^T\vartheta(z)+\varepsilon
\label{crit}
\end{equation}
where, $W \in \mathbb{R}^N$ ($N$ being the size of the regressor vector) denotes the ideal weight vector that can closely approximate the value function. And, $\vartheta(z)=[\vartheta_1(z),\vartheta_2(z),...,\vartheta_{N}(z)]^T \in \mathbb{R}^{N}$ represents a set of regressor functions, with following properties such as: $\vartheta_j(z) \in C^1$ and $\vartheta_j(0)=0$ and $\vartheta_j$s are linearly independent of each other.
Substituting (\ref{crit}) in (\ref{opt_u}),
\begin{equation}
\begin{split}
u^*(z)=-u_m\tanh{\Big(\frac{1}{2u_m}R^{-1}{G}(z)^T \nabla{\vartheta}^TW+\varepsilon_{uu}\Big)}
\end{split}
\label{ctr_crit}
\end{equation}
where, $\varepsilon_{u^*}=(1/2u_m)R^{-1}{G}^T(z)\nabla{\varepsilon}\in \mathbb{R}^m$. Next, substituting (\ref{crit}) in (\ref{hjb_sum}), the HJB equation can be written as,
\begin{equation}
\begin{split}
W^T\nabla{\vartheta}{F}(z)-\gamma W^T\vartheta+z^TQ_1z+d_M^2+\\
u_m^2\sum_{i=1}^{m}\ln[1-\tanh^2{(\tau_{1i}+\epsilon_{u^*_i})}]+\nabla\epsilon^T{F}(z)=0
\end{split}
\label{res_err}
\end{equation}
where, $\tau_{1}=(1/2u_m)R^{-1}{G}(z)^T \nabla{\vartheta}^TW=[\tau_{11},...,\tau_{1m}]^T \in \mathbb{R}^m$,  $\varepsilon_{u^*}=[\varepsilon_{u^*_{11}},\varepsilon_{u^*_{12}},...,\varepsilon_{u^*_{1m}}]^T$.
Upon using Mean value theorem \cite{rudin1964principles}, Eq.  (\ref{res_err}) becomes:

\begin{equation}
\begin{split}
W^T\nabla{\vartheta}{F}(z)-\gamma W^T\vartheta+z^TQ_1z+d_M^2+\\
u_m^2\sum_{i=1}^{m}\ln[1-\tanh^2{(\tau_{1i})}]+\epsilon_{HJB}=0
\end{split}
\label{hjb_ideal}
\end{equation}
where, $\epsilon_{HJB}$ represents the HJB approximation error \cite{abu2005nearly},\cite{modares2014integral} having a form similar to the one in \cite{liu2015reinforcement} and is given as,
\begin{equation}
\epsilon_{HJB}=\nabla\epsilon^T{F}(z)+\sum_{i=1}^{m}\frac{2u_m^2}{p_{1i}}\tanh{p_{2i}}(\tanh^2{p_{2i}}-1)\epsilon_{u^*_i}
\end{equation}
where, $p_{1i}\in \mathbb{R}$ and $p_{2i} \in \mathbb{R}$ considered between $1-\tanh^2{A_i(z)}$ and $1-\tanh^2{\tau_i}$. Now, using (\ref{crit}) and mean value theorem, the optimal control can be re-written as:
\begin{equation}
u^*=-u_m\tanh{(\tau_1(z))}+\epsilon_{u}    
\label{ustar1}
\end{equation}
where $\tau_{1}(z)=(1/2u_m)R^{-1}\hat{G}^T \nabla{\vartheta}^TW=[\tau_{11},...,\tau_{1m}]^T \in \mathbb{R}^m$ and 
$\epsilon_{u}=-(1/2)((I_m-diag(\tanh^2{(q)}))R^{-1}\hat{G}^T\nabla{\epsilon})$ 
with $q \in \mathbb{R}^m$ and $q_i \in \mathbb{R}$ considered between $\tau_{1i}+\varepsilon_{uui}$ and $\epsilon_{uui}$ i.e., $i^{th}$ element of $\tau_{1}+\varepsilon_{uu}$ and $\epsilon_{uu}$, respectively such that tangent of $\tanh{(q)}$ is equal to the slope of the line joining $\tanh{(\tau_1+\epsilon_{uu})}$ and $\tanh{\varepsilon_{uu}}$.  
For the detailed proof, refer to Lemma \ref{mean_val_lem}. 
In the subsequent analysis, $\tau_1 \triangleq \tau_1(z)$. 

Since ideal weights that can accurately approximate the value function are unknown, their estimates will be used instead as follows.
\begin{equation}
{V}(z)=\hat{W}^T\vartheta(z) 
\label{est_weight}
\end{equation}
Error in critic weights is given by $\tilde{W}=W-\hat{W}$. Using (\ref{est_weight}) the estimated optimal control action can be described as:
\begin{equation}
 \hat{u}(z)=-u_m\tanh{\Bigg(\frac{1}{2u_m}{G}^T(z)\nabla{\vartheta}^T\hat{W}\Bigg)}
 \label{est_ctrol}
\end{equation}
From (\ref{hjb_sum}) and (\ref{est_weight}) the HJB approximation error is obtained as follows.
\begin{equation}
\begin{split}
\hat{H}(z,\hat{W})& =\hat{W}^T\nabla{\vartheta}{F}(z)-\gamma \hat{W}^T\vartheta+z^TQ_1z+d_M^2+\\
& u_m^2\sum_{i=1}^{m}\ln[1-\tanh^2{(\tau_{2i})}]\triangleq e(z,\hat{W}) 
\end{split}
\label{est_ha}
\end{equation}
where, $e(z,\hat{W})$ is the HJB error (referred to as $\hat{e}$ in subsequent discussion) and $\tau_{2}(z)=(1/2u_m){G}^T(z)\nabla{\vartheta}^T\hat{W}=[\tau_{21}(z),...,\tau_{2m}(z)]^T \in \mathbb{R}^m$. Next, from (\ref{hjb_ideal}) and (\ref{est_ha}) the HJB error can be expressed in terms of $(\tilde{W}$ which is $W-\hat{W})$ and $W$ as \cite{liu2015reinforcement}:
\begin{equation}
\small
e=-\tilde{W}^T\nabla{\vartheta}{F}(z)+\gamma\tilde{W}^T\vartheta+\sum_{i=1}^mu_m^2[\Gamma(\tau_{2i})-\Gamma(\tau_{1i})]-\epsilon_{HJB}
\label{er}
\end{equation}
where, $\Gamma(\tau_{\iota i})=\ln[1-\tanh^2{\iota i}]$, $\iota=1,2$. It is observed that for all $\tau_{\iota i}(z) \in \mathbb{R}$, $\Gamma(\tau_{\iota i})$ can be represented by:
\begin{equation}
\Gamma(\tau_{\iota i})=-2\ln[1+exp(-2\tau_{\iota i}sgn(\tau_{\iota i}))]-2\tau_{\iota i}sgn(\tau_{\iota i})+\ln(4)
\end{equation}
where, $sgn$ is signum function. Also note that:
\begin{equation}
\begin{split}
\sum_{i=1}^m\Gamma(\tau_{\iota i})&=-2\sum_{i=1}^m\ln[1+exp(-2\tau_{\iota i}sgn(\tau_{\iota i}))]-\\
&2\tau^T_{\iota}sgn(\tau_{\iota})+m\ln(4)
\end{split}
\label{sumtau}
\end{equation}
Therefore, using (\ref{er}) and (\ref{sumtau}), $e$ in terms of $\tilde{W}$, is obtained as \cite{liu2015reinforcement}:
\begin{equation}    
\begin{split}
\hat{e}&=2u_m^2[\tau^T_1sgn(\tau_1)-\tau^T_2sgn(\tau_2)]-\tilde{W}\nabla{\vartheta}{F}(z)+u_m^2\Delta\tau\\
& -\epsilon_{HJB} \\
&=u_m[W^T\nabla{\vartheta}{G}(z)sgn(\tau_1(z))-\\
& \hat{W}^T\nabla{\vartheta}{G}(z)sgn(\tau_2(z))]\tilde{W}^T\nabla{\vartheta}{F}(z)+u_m^2\Delta\tau-\epsilon_{HJB} \\
&=-\tilde{W}^T[\nabla{\vartheta}{F}(z)-u_m\nabla{\vartheta}{G}(z)sgn(\tau_2)]+\rho(z)
\end{split}
\label{e_w_til}
\end{equation}
where, 
\begin{equation}
\begin{split}
\Delta\tau&=2\sum_{i=1}^m\ln\Bigg(\frac{1+exp[-2\tau_{1i}(z)sgn(\tau_{1i}(z))]}{1+exp[-2\tau_{2i}(z)sgn(\tau_{2i}(z))]}\Bigg) \\
\rho(z)&=u_m W^T\nabla{\vartheta}{G}(z)[sgn(\tau_1(z)-sgn(\tau_2(z)))]+u_m^2\Delta\tau\\
&-\epsilon_{HJB}
\end{split}
\end{equation}

\subsection{Existing update laws in literature}\label{param_update}
In traditional RL literature for continuous time nonlinear systems, a quadratic cost function of the form, $E=(1/2)\hat{e}^2$ is chosen, and then gradient descend (GD) is used to drive the parameters $\hat{W}$ so as to minimize this cost $E$ and thus to minimize the HJB error. The following tuning law has been proposed in \cite{vamvoudakis2010online},\cite{bhasin2013novel},\cite{yang2014reinforcement},\cite{zhang2011data},\cite{modares2014optimal},\cite{lewis2013reinforcement}.
\begin{equation}
\dot{\hat{W}}=-\frac{\alpha}{(1+\phi^T\phi)^2}\frac{\partial E}{\partial \hat{W}}=-\frac{\alpha\phi}{(1+\phi^T\phi)^2}\hat{e}
\label{gd}
\end{equation}
where, $\phi=\nabla{\vartheta}({F}(z)+{G}(z)\hat{u})$, $\alpha>0$ is the learning rate, and $1+\phi^T\phi$ is the normalization factor. Then in 2015, Yang et al. \cite{liu2015reinforcement} proposed a modified version of (\ref{gd}) for optimal regulation problems wherein they used constant learning rate in their gradient descent formulation. Their update mechanism was given as below.
\begin{equation}
\begin{split}
\dot{\hat{W}}&=-\alpha\bar{\phi}\Bigg(Y(x)+ d_M^2(x)+u_m^2\sum_{i=1}^m\ln[1-\tanh^2{(\tau_{2i}(x))}]\Bigg)\\
& +\frac{\alpha}{2}\Xi(x,\hat{u})\nabla{\vartheta}{G}(x)[I_m-\mathcal{B}(\tau_{2}(x))]{G}^T(x)L_{2x}\\
& +\alpha\Big((K_1\varphi^T-K_2)\hat{W}+u_m\nabla{\vartheta}{G}(x)[\tanh{(\tau_2(x))}-\\
& sgn(\tau_{2}(x))]\frac{\varphi^T}{m_s}\hat{W} \Big)
\end{split}
\label{tuning_law}
\end{equation}
where, $x$ is the actual state of the system (not the augmented state), $\alpha>0$, $\phi=\nabla{\vartheta}({F}(x)+{G}(x)\hat{u})$, $\bar{\phi}=\phi/m_s^2$, $\varphi=\phi/m_s$, $m_s=1+\phi\phi^T$, $Y(x)=\hat{W}^T\nabla{\vartheta}{F}+x^TQ_1x$, $\mathcal{B}=diag\{\tanh^2{(\tau_{2i}(x))}\},~i=1,2...,m$.  

\section{Variable gain-based update law}\label{sec:var_gain}
\subsection{Update law}
It can be observed in  \cite{yang2015robust} and \cite{liu2015reinforcement} that significantly high amount of time is taken by the approximate optimal controller to bring the states \cite{liu2015reinforcement} or the error in states ($x-x_d$) \cite{yang2015robust} to a small residual set around origin. In both the above papers, a smaller learning rate was selected to avoid oscillations. However, small values of learning rate results in longer learning phase.
In order to address this issue, in this paper, a tuning law with variable learning rate gradient descent is proposed and expressed as follows.
\begin{equation}
\begin{split}
\dot{\hat{W}}&=-\alpha(|e(z,\hat{W})|^{k_2}+l)\bar{\phi}e(z,\hat{W})\\
& +\frac{\alpha}{2}\Xi(z,\hat{u})\nabla{\vartheta}{G}(z)[I_m-\mathcal{B}(\tau_{2}(z))]{G}^T(z)L_{2z}\\
& +\alpha(|e(z)|^{k_2}+l)\Big((K_1\varphi^T-K_2)\hat{W}\\
&+u_m\nabla{\vartheta}{G}(z)[\tanh{(\tau_2(z))}-sgn(\tau_{2}(z))]\frac{\varphi^T}{m_s}\hat{W} \Big)
\end{split}
\label{tuning_law1}
\end{equation}
where, $\alpha>0$ is the learning rate, $l$ is a small positive constant, and $e(z,\hat{W})$ is the HJB error as mentioned in (\ref{est_ha}). 
In the subsequent analysis, $g_1\triangleq |\hat{e}|^{k_2}+l$ further, to ease the development of stability proof, $k_2=1$, however it can be set to any positive value.
In (\ref{tuning_law1}), the term $\phi$ is defined as $\phi=\nabla{\vartheta}({F}(z)+{G}(z)\hat{u})-\gamma\vartheta(z)$, $\bar{\phi}=\phi/m_s^2$, $\varphi=\phi/m_s$, $m_s=1+\phi\phi^T$, 
$\mathcal{B}=diag\{\tanh^2{(\tau_{2i}(z))}\},~i=1,2...,m$.
The term $\Xi(z,\hat{u})$ is a piece-wise continuous indicator function defined as in \cite{liu2015reinforcement}.
\begin{equation}
\Xi(z,\hat{u})=  \begin{cases} 
      0, & if~ \Sigma(z(t)) <0 
      \\
      1, & otherwise
   \end{cases}
\end{equation}
where, $\Sigma(z(t))=L^T_{2z}({F}(z)+{G}(z)\hat{u})$ denotes the rate of variation of Lyapunov function along the system trajectories. 
It is to be noted that, $L_2=(1/2)z^Tz$ and hence $L_{2z}=z$.
The constants, $k_2>0$ provide an augmentation to the controller by enabling accelerated learning, when the HJB error ($e(z,\hat{W})$) is large. 
On the other hand, it dampen the learning process when $e(z,\hat{W})$ diminishes to a small quantity.
Proper choice of this constants allows for the use of higher value of learning rate without significant oscillations as will be observed in the simulation results presented in Section \ref{result}. 
Thus, the controller can bring the error within a small residual set around origin much quickly without any significant oscillations.

Note that the form of (\ref{tuning_law1}) is different from (\ref{tuning_law}) that was presented in literature  \cite{liu2015reinforcement} in following ways. 
\begin{itemize}
\item The $\phi$ in (\ref{tuning_law1}) has an additional term $\gamma\vartheta(z)$ and $e(z,\hat{W})$ has $-\gamma \hat{W}\vartheta(z)$. Both these terms arise because of the discounted cost function that was used in (\ref{tuning_law1}) compared to (\ref{tuning_law}).
\item  
The variable gain in first term of (\ref{tuning_law1}) is chosen to be a function of HJB error. 
This has been done in order to accelerate the reduction of HJB error when it is large and dampen the reduction process when the HJB error becomes small. 
The added benefit of variable gain is that it shrinks the size of the residual sets for both error in state and error in parameter as will become clear in the stability proof.
\item The second term in (\ref{tuning_law1}) is dependent on the variation of Lyapunov function along the system trajectories. It is $0$, when the Lypunov function is strictly decreasing along the system trajectories as shown by the piece-wise indicator function $\Xi(z,\hat{u})$. 
However it comes into effect when the Lyapunov function is non-decreasing along the system trajectories. 
It implies that the control action generated at any time step during policy improvement leads to growth in Lyapunov function along the augmented system trajectories. 
The second terms starts pulling the critic weights in the direction where the Lyapunov is no more increasing along the system trajectories. 
In order to fully understand it, let $\Sigma$ denote the variation of Lyapunov function along the sytem trajectories as $\Sigma=L_{2z}({F}(z)-u_m{G}(z)\tanh{\tau_{2}(z)})$
\begin{itemize}
    \item Gradient descent is utilized in \cite{liu2015reinforcement} to drive the weights in direction such that $\Sigma$ can be reduced and eventually made negative.
\end{itemize}
\begin{equation}
\begin{split}
-\alpha\frac{\partial \Sigma}{\partial \hat{W}}&=-\alpha\frac{\partial [L_{2z}({F}(z)-u_m{G}(z)\tanh{\tau_{2}(z)})]}{\partial\hat{W}} \\
&=\alpha\Bigg(\frac{\partial \tau_{2}(z)}{\partial\hat{W}}\Bigg)^T\frac{\partial[u_m L_{2z}{G}(z)\tanh{\tau_{2}(z)}]}{\partial\hat{W}} \\
&=\frac{\alpha}{2}\nabla{\vartheta}{G}(z)[I_m-\mathcal{B}(\tau_{2}(z))]{G}^T(z)L_{2z}
\end{split}
\end{equation}


\item The last term in (\ref{tuning_law1}) provides control over the UUB sets as mentioned in \cite{liu2015reinforcement}. 
Proper choice of gains $K_1$ and $K_2$ can shrink the UUB ball close to the origin.
\end{itemize}

Using (\ref{e_w_til}) and (\ref{tuning_law1}) the dynamics of error in critic weights is then given as,
\begin{equation}
\footnotesize
\begin{split}
\dot{\tilde{W}}&=\alpha g_1\frac{\varphi}{m_s}\Big[-\tilde{W}^T\phi+u_m\tilde{W}^T\nabla{\vartheta}{G}(z)\mathcal{F}(z)+\rho(z)\Big]\\
& -\frac{\alpha}{2}\Xi(z,\hat{u})\nabla{\vartheta}{G}(z)\Big[I_m-\mathcal{B}(z)\Big]{G}^T(z)L_{2z} \\
&+\alpha g_1\Big[\nabla{\vartheta}{G}(z)\mathcal{F}(z)\frac{\varphi^T}{m_s}\hat{W}+(K_2-K_1\varphi^T)\hat{W}\Big]
\end{split}
\label{crit_err}
\end{equation}
where, $\mathcal{F}(z)=sgn(\tau_2(z))-\tanh{(\tau_2(z))}$. 
\vspace{0.05cm}
\subsection{Proof of Stability of Online Tuning Law}\label{stab}
\begin{assumption}\label{weight_as}
\textnormal{Ideal NN weight vector $W$ is considered to be bounded, i.e., $\|W\| \leq  W_M$. There exists positive constants $b_{\epsilon}$ and $b_{\epsilon z}$ that bound the approximation error and its gradient such that $\|\varepsilon(z)\| \leq b_{\epsilon}$ and $\|\nabla{\varepsilon}\| \leq b_{\epsilon z}$. 
This is in line with Assumptions 3b of \cite{vamvoudakis2010online},  Assumption 5 of \cite{liu2015reinforcement} and Assumptions made in Section 4.1 in \cite{modares2014optimal}}.
\end{assumption}

\begin{assumption}\label{regres_as}
\textnormal{Critic regressors are considered to be bounded as well: $\|\vartheta(z)\| \leq b_{\vartheta}$ and $\|\nabla{\vartheta}(z)\| \leq b_{\vartheta z}$. This is in line with Assumption 4 of \cite{yang2015robust}, Assumption 6 of \cite{liu2015reinforcement} and Assumption 4 of \cite{modares2014optimal}}.
\end{assumption}
In this paper, both the assumptions hold true because, there exists a stabilizing term (second term) in the update law (\ref{tuning_law1}) that comes into effect when Lyapunov function starts growing along the system trajectories. 
This term helps in pulling the system out of region where Lyapunov function is growing thus ensuring that the trajectories remain finite within a region $\Omega_1 \subset R^{2n}$.
\begin{assumption}\label{lyapunov_as}
\textnormal{Let $L_2 \in C^1$ be a continuously differentiable and radially unbounded Lyapunov candidate for (\ref{nad}) and satisfies $\dot{L}_2=L_{2z}^T({F}(z)+{G}u^*)<0$. Furthermore, a symmetric and positive definite $\Lambda(z) \in \mathbb{R}^{n\times n}$ can be found, such that, $L_{2z}^T({F}(z)+{G}u^*)=-L^T_{2z}\Lambda(z) L_{2z}$, where $L_{2z}$ is the partial derivative of $L_{2}$ wrt $z$.
In the subsequent analysis, $\Lambda \triangleq \Lambda(z)$.}
\end{assumption}

Following Lipschitz continuity of $({F}(z)+{G}u^*)$ in $z$, this assumption can be shown reasonable. It is also in line with Assumption 4 mentioned in \cite{liu2015reinforcement} and \cite{dierks2010optimal}.

\begin{theorem}
Let the CT nonlinear augmented system be described by (\ref{nad}) with associated HJB as (\ref{hjb_sum}) and approximate optimal control as (\ref{est_ctrol}) and let the Assumptions \ref{uncert}-\ref{lyapunov_as} hold true, then the tuning law (\ref{tuning_law1}) makes $L_{2z}$ and $\tilde{W}$ uniform ultimate boundedness (UUB) stable.
Further, the UUB set could be made arbitrary small by proper selection of gains $K_1, K_2$ and exponent $k_2$ in \eqref{tuning_law1}.
\end{theorem}

\begin{proof}
Let the Lyapunov candidate be $L=L_2+(1/2\alpha)\tilde{W}^T\tilde{W}$ (Where $L_2$ is a positive definite function of augmented state as defined after (\ref{tuning_law1})). 
Derivative of $L$ w.r.t. time is obtained as follows.
\begin{equation}
\begin{split}
\dot{L}&=L_{2z}({F}(z)+{G}(z)\hat{u})+\dot{\tilde{W}}\alpha^{-1}\tilde{W} \\
&=L_{2z}({F}(z)-u_m{G}(z)\tanh{(\tau_2(z))})+\dot{\tilde{W}}\alpha^{-1}\tilde{W}
\end{split}
\end{equation}
Utilizing error dynamics of weights, i.e (\ref{crit_err}) and using the fact that $\dot{z}=F(z)+G(z)\hat{u}$, the last term of  Lyapunov derivative becomes:
\begin{equation}
\footnotesize
\begin{split}
&\dot{\tilde{W}}\alpha^{-1}\tilde{W}=\Big[-\tilde{W}^T\phi+u_m\tilde{W}^T\nabla{\vartheta}{G}(z)\mathcal{F}(z)+\rho(z)\Big]g_1\frac{\varphi^T}{m_s}\tilde{W}\\
& -\frac{1}{2}g_2\Xi(z,\hat{u})L^T_{2z}{G}(z)\Big[I_m-\mathcal{B}(\tau_2(z))\Big]{G}^T(z)\nabla{\vartheta}^T\tilde{W} \\
&+g_1u_m\tilde{W}\nabla{\vartheta}{G}(z)\mathcal{F}(z)\frac{\varphi^T}{m_s}\hat{W}+g_1\tilde{W}^T(K_2\hat{W}-K_1\varphi^T\hat{W}) \\
&=-g_1\tilde{W}\varphi\varphi^T\tilde{W}+g_1\delta(z)\varphi^T\tilde{W}+g_1\tilde{W}^T\beta(z)+g_1\tilde{W}^T(K_2\hat{W}\\
&-K_1\varphi^T\hat{W}) \underbrace{-\frac{1}{2}\Xi(z,\hat{u})L^T_{2z}{G}[I_m-\mathcal{B}(\tau_2(z))]{G}^T(z)\nabla{\vartheta}^T\tilde{W}}_\text{$S$}
\end{split}
\label{ldot4}
\end{equation}
where $\delta(z)\triangleq \rho(z)/m_s$ and $\beta(z)\triangleq u_m\nabla{\vartheta}{G}(z)\mathcal{F}(z)(\varphi^T/m_s)W$.
Let, $A\triangleq-g_1\tilde{W}\varphi\varphi^T\tilde{W}+g_1\delta(z)\varphi^T\tilde{W}+g_1\tilde{W}^T\beta(z)\\
+g_1\tilde{W}^T(K_2\hat{W}-K_1\varphi^T\hat{W})$.
The last term in (\ref{ldot4}) can be expressed as:
\begin{equation}
\begin{split}
 \tilde{W}^T(K_2\hat{W}-K_1\varphi^T\hat{W})&=\tilde{W}^TK_2W-\tilde{W}^TK_2\tilde{W}-\\
 &\tilde{W}^TK_1\varphi^TW+\tilde{W}^TK_1\varphi^T\tilde{W}
\end{split}
\end{equation}
Let, 
\begin{equation}
\mathcal{J}\triangleq[\tilde{W}^T\varphi,\tilde{W}^T]^T
\label{j_def}
\end{equation}
then (\ref{ldot4}) can be re-written as:
\begin{equation}
\begin{split}
\dot{\tilde{W}}^T\tilde{W}/\alpha=A+S\leq g_1(-\lambda_{min}(M)\|\mathcal{J}\|^2+b_N\|\mathcal{J}\|)+S
\end{split}
\label{eq:wtildedot}
\end{equation}

where, $M$ and $N$ are defined as:
\begin{equation}
\begin{split}
M=\begin{pmatrix}
1    &      -\frac{1}{2}K^T_1 \\
-\frac{1}{2}K_1  &  K_2
\end{pmatrix};~N=\begin{pmatrix}
 \delta(z)  \\
 (\beta(z)+K_2W-K_1\varphi^TW)
\end{pmatrix}
\end{split}
\label{eq:MN}
\end{equation}
where, $b_N$ is the upper bound of N which is given by the expression:
\begin{equation}
\|N\| \leq b_N= \max(\|N\|)
\label{bn}
\end{equation}
In (\ref{eq:MN}), if $K_1$ and $K_2$ are chosen such that $K_2$ is symmetric, then $M$ becomes symmetric.
Further, in order to ensure that $\lambda_{min}(M)$ is real and positive, $K_1$ and $K_2$ should be selected such that $M$ is positive definite.
Further, $A$ can be developed by leveraging $g_1$ as a function of $\tilde{W}$.
From (\ref{e_w_til}), $g_1$ as a function of $\tilde{W}$ is, 
$g_1=|\hat{e}(\tilde{W})|^{k_2}+l$ (with $k_2=1$).
\begin{equation}
\begin{split}
g_1&=|-\tilde{W}^T\phi+u_m\tilde{W}^T\nabla{\vartheta}{G}(z)\mathcal{F}(z)+\rho(z)|+l\\
&\leq \|\rho\|+\|\tilde{W}\|\|\phi\|+u_m\|\tilde{W}\|b_{\vartheta z}g_M2\sqrt{m}+l\\
&\leq \|\tilde{W}\|\underbrace{(\|\phi\|+u_mb_{\vartheta z}g_M2\sqrt{m})}_\text{$A_1$}+\underbrace{(\|\rho\|+l)}_\text{$A_2$}
\end{split}
\end{equation}
where, $\|\mathcal{F}\| \leq 2\sqrt{m}$.
It could be noted that $\delta(z)=\rho(z)/m_s$ is one of the component of vector $N$, and by appropriately selecting $K_1$ and $K_2$ in $N$, and selecting a very small offset $l$, it can be ensured that $A_2=\|\rho\|+l \leq b_N$. 
Also, from (\ref{j_def}), $\|\tilde{W}\| \leq \|\mathcal{J}\|$ 
, therefore,
\begin{equation}
\begin{split}
g_1\leq A_1\|\mathcal{J}\|+b_N
\end{split}
\label{eq:g1}
\end{equation}

Therefore, the Lyapunov derivative can be rendered in the following inequality:
\begin{equation}
\begin{split}
\dot{L} &\leq L_{2z}({F}(z)+{G}(z)\hat{u})+\\
&(A_1\|\mathcal{J}\|+b_N)(-\lambda_{min}(M)\|\mathcal{J}\|^2+b_N\|\mathcal{J}\|)+S
\end{split}
\label{ldot}
\end{equation}

Based on the variation of Lyapunov function along the system trajectories, which is captured by the value of the piecewise continuous function, $\Xi(z,\hat{u})$, (\ref{ldot}) can be explained in two cases:\\
\textbf{Case(i)}:
When $\Xi(z,\hat{u})=0\Rightarrow S=0$.\\
By definition, in this case, $L_{2z}^T\dot{z}<0$ (where $\dot{z}={F}(z)+{G}(z)\hat{u}$). 
Therefore, 
\begin{equation}
\begin{split}
\dot{L} & \leq L^T_{2z}\dot{z}+\underbrace{(A_1\|\mathcal{J}\|+b_N)(-\lambda_{min}(M)\|\mathcal{J}\|^2+b_N\|\mathcal{J}\|)}_\text{$A$}
\end{split}
\label{ldot2}
\end{equation}
In order for $\dot{L}$ to be negative definite, $A$ should be negative, now for $\|\mathcal{J}\|\neq 0$, $A<0$ when,
\begin{equation}
\footnotesize
\begin{split}
-&A_1\lambda_{min}(M)\|\mathcal{J}\|^2+\underbrace{(b_NA_1-\lambda_{min}(M)b_N)}_\text{$\triangleq B_1$}\|\mathcal{J}\|+b_N^2<0\\
\Rightarrow&\|\mathcal{J}\|>\frac{B_1}{2A_1\lambda_{min}(M)}+\sqrt{\frac{B_1^2}{4A_1^2\lambda^2_{min}(M)}+\frac{b_N^2}{A_1\lambda_{min}(M)}}\\
\Rightarrow&\|\mathcal{J}\|>\frac{b_N}{\lambda_{min}(M)}\underbrace{\left[\frac{1}{2}\left(1-\gamma_1\right)+\sqrt{\frac{1}{4}\left(1-\gamma_1\right)^2-\gamma_1}\right]}_\text{$\triangleq \Gamma$}
\end{split}
\label{eq:S}
\end{equation}
where, $\gamma_1 \triangleq \frac{\lambda_{min}(M)}{A_1}$, therefore, if, $0 \leq \gamma_1 \leq 3-\sqrt{8}\approx0.17$, then $.478\leq \Gamma \leq 1$.
Also, recall from the definition of $\mathcal{J}$ in (\ref{j_def}), the upper bound of $\|\mathcal{J}\|$ can be obtained as,
\begin{equation}
 \|\mathcal{J}\| \leq \Big(\sqrt{1+\|\varphi\|^2}\Big)\|\tilde{W}\|
 \label{Y-ineq}
\end{equation}
Therefore, from lower and upper bounds of $\mathcal{J}$ in (\ref{eq:S}) and (\ref{Y-ineq}), respectively, the bound over $\|\tilde{W}\|$ becomes,  
\begin{equation}
\|\tilde{W}\| > \frac{\frac{b_N}{\lambda_{min}(M)}\Gamma}{\sqrt{1+\|\varphi\|^2}} 
\label{ld1}
\end{equation}

It could be seen that $\tilde{W}$ is UUB stable with bound given in the RHS of (\ref{ld1}). 
Also, note that if Eq. (\ref{ld1}) holds, then the negative definiteness of $\dot{L}$ is ensured.
 

\textbf{Case(ii)}: If $\Xi(z,\hat{u})=1$ \\
By definition, in this case, the Lyapunov function is non-decreasing along the system trajectories.
The analysis of this case follows similarly as in the previous one, except, the last term in the right hand side (RHS) of (\ref{ldot}) also needs to be considered. For that, (\ref{opt_u}), (\ref{ldot}) and Assumption \ref{lyapunov_as} would be utilized.
\begin{equation}
\footnotesize
\begin{split}
\dot{L} & \leq L^T_{2z}{F}(z)-u_m L^T_{2z}{G}(z)\Big[\tanh{(\tau_2(z))}+\frac{2}{2u_m}[I_m\\
&-\mathcal{B}(\tau_2(z))]{G}^T\nabla{\vartheta}^T\tilde{W}\Big]+(A_1\|\mathcal{J}\|+b_N)(-\lambda_{min}(M)\|\mathcal{J}\|^2+b_N\|\mathcal{J}\|) \\
\end{split}
\end{equation}
Now, adding and subtracting $L^T_{2z}({G}(z)u^*)$ one gets:
\begin{equation}
\footnotesize
\begin{split}
\dot{L} & \leq L^T_{2z}({F}(z)+{G}u^*)-u_m L^T_{2z}{G}(z)\Big[\tanh{(\tau_2(z))}\\
& +\frac{g_2}{2u_m}[I_m-\mathcal{B}(\tau_2(z))]{G}^T\nabla{\vartheta}^T\tilde{W}\Big]+(A_1\|\mathcal{J}\|+b_N)\\
&\times(-\lambda_{min}(M)\|\mathcal{J}\|^2+b_N\|\mathcal{J}\|)-L^T_{2z}{G}(z)(-u_m\tanh{(\tau_1(z))}+\epsilon_{u})
\end{split}
\label{ldot_2case}
\end{equation}
Using the inequality $\|\tanh{(\tau_1(z))}-\tanh{(\tau_2(z))}\| \leq T_m$ (see Lemma \ref{tanhlem} in Appendix \ref{L1}), Assumption \ref{weight_as}, \ref{regres_as} and \ref{lyapunov_as},  Inequality (\ref{ldot_2case}) can be re-written as:
\begin{equation}
\footnotesize
\begin{split}
\dot{L} 
& \leq -L^T_{2z}\Lambda L_{2z}-L^T_{2z}{G}(z)\epsilon_{u}+u_m \|L^T_{2z}\|g_M\|\tanh{(\tau_1(z)}\\
& -\tanh{\tau_2(z)}\|+(A_1\|\mathcal{J}\|+b_N)(-\lambda_{min}(M)\|\mathcal{J}\|^2+b_N\|\mathcal{J}\|) \\
& +\frac{g_2}{2}\|L^T_{2z}\|\|\mathcal{N}_1\nabla^T{\vartheta}\tilde{W}\| \\
& \leq -\lambda_{min}(\Lambda)\|L_{2z}\|^2+\|L_{2z}\|(T_mu_m g_M+\frac{g_2}{2}\|\mathcal{N}_1\nabla^T{\vartheta}\tilde{W}\|)\\
&+(A_1\|\mathcal{J}\|+b_N)(-\lambda_{min}(M)\|\mathcal{J}\|^2+b_N\|\mathcal{J}\|)+\frac{1}{2}\|L^T_{2z}\|g_M^2b_{\epsilon z}\\
\end{split}
\label{first_ineq}
\end{equation}
where, $\mathcal{N}_1\triangleq{G}(z)[\mathcal{B}(\tau_2(z))-I_m]{G}^T(z)$. 
Now, two positive constant numbers $n_1$ and $n_2$ are defined such that $n_1+n_2=1$. 
In the following analysis, $\|\tilde{W}\|^2 \leq \|\mathcal{J}\|^2$ is also utilized. 
Therefore the inequality in (\ref{first_ineq}) can be developed as follows:
\begin{equation}
\footnotesize
\begin{split}
\dot{L}&\leq -n_1\lambda_{min}(\Lambda)\|L^T_{2z}\|^2+\|L_{2z}\|T_mu_m g_M+\frac{\|g_2/2\mathcal{N}_1\nabla^T\vartheta\|^2\|\tilde{W}\|^2}{4n_2\lambda_{min}(\Lambda)}\\
&-n_2\lambda_{min}(\Lambda)\Big(\|L^T_{2z}\|-\frac{\|g_2/2\mathcal{N}_1\nabla^T\vartheta\|\|\tilde{W}\|}{2n_2\lambda_{min}(\Lambda)}\Big)^2+A\\
&\leq -\underbrace{n_1\lambda_{min}(\Lambda)}_\text{$\triangleq a$}\|L^T_{2z}\|^2+\|L_{2z}\|\underbrace{(T_mu_m g_M+\frac{1}{2}g_M^2b_{\epsilon z})}_\text{$\triangleq b$}\\
&+\underbrace{\left(\frac{\|g_2/2\mathcal{N}_1\nabla^T\vartheta\|^2}{4n_2\lambda_{min}(\Lambda)}\right)}_\text{$\triangleq c$}\|\mathcal{J}\|^2+(A_1\|\mathcal{J}\|+b_N)(-\lambda_{min}(M)\|\mathcal{J}\|^2+b_N\|\mathcal{J}\|)\\
&\leq-a\left(L_{2z}-\frac{b}{2a}\right)^2+\frac{b^2}{4a}+\|\mathcal{J}\|\Big(-A_1\lambda_{min}(M)\|\mathcal{J}\|^2\\
&+(A_1b_N-b_N\lambda_{min}(M)+c)\|\mathcal{J}\|+b_N^2\Big)
\end{split} 
\label{ldot_semi_final}
\end{equation}
In order for $\dot{L}$ to be negative definite, 
\begin{equation}
\begin{split}
& -a\left(L_{2z}-\frac{b}{2a}\right)^2+\frac{b^2}{4a}<0 ~\Rightarrow \|L_{2z}\| > \frac{b}{a} 
\end{split}
\label{l2z1}
\end{equation}
and
\begin{equation}
\footnotesize
\begin{split}
& \|\mathcal{J}\|\Big(-A_1\lambda_{min}(M)\|\mathcal{J}\|^2+(A_1b_N-b_N\lambda_{min}(M)+c)\|\mathcal{J}\|\\
&+b_N^2\Big)<0\\
\Rightarrow&\|\mathcal{J}\|>\frac{b_N}{\lambda_{min}(M)}\underbrace{\left[\frac{1}{2}\left(1-\gamma_1+\alpha_2\right)+\sqrt{\left(\frac{1}{2}(1-\gamma_1+\alpha_2)\right)^2+\gamma_1}\right]}_\text{$\triangleq \Gamma^{'}$}
\end{split}
\label{Yineq1}
\end{equation}
where, $\alpha_2=c/(A_1b_N)\geq 0$ and $\Gamma^{'}$ is a fractional scaling factor that can scale the term $\frac{b_N}{\lambda_{min}(M)}$.
Now, in order for $\Gamma^{'}$ to lie between $[c_1,c_2]$ such that $\frac{1}{2}<c_1<c_2<1$, $\gamma_1$ must lie between,
\begin{equation}
\frac{2c_2(1+\alpha_2)-c_2^2}{2c_2-1}\leq \gamma_1\leq \frac{2c_1(1+\alpha_2)-c_1^2}{2c_1-1}
\end{equation}
From (\ref{Y-ineq}) and (\ref{Yineq1}), UUB set for $\tilde{W}$ is,
\begin{equation}
\begin{split}
\|\tilde{W}\|>\frac{\frac{b_N}{\lambda_{min}(M)}}{\sqrt{1+\|\varphi\|^2}}\Gamma^{'}
\end{split}
\label{bound}
\end{equation}
Therefore, for $\dot{L}$ to be negative definite, both \eqref{l2z1} and \eqref{bound} should hold true.
\vspace{-.1cm}

This completes the stability proof of the update mechanism (\ref{tuning_law1}).
\end{proof}


\begin{remark}
\textnormal{Note that for \textbf{Case (i)}, if $\tilde{W}$ satisfies (\ref{ld1}), and for \textbf{Case (ii)}, if $L_{2z}$ and $\tilde{W}$ satisfy (\ref{l2z1}) and (\ref{bound}), respectively, then it leads to decreasing $\tilde{W}$ and $L_{2z}$.
It is evident that when the Lyapunov function is decreasing along the augmented state trajectory, variable learning rate has a direct influence over UUB bound for error in critic NN weights ($\tilde{W}$). 
By suitable selection of $K_1\text{ and }K_2$, the scaling factor $\Gamma$ in (\ref{eq:S}) can be varied between $.478$ and $1$ or $\Gamma^{'}$ in (\ref{Yineq1}) can be varied between $\frac{1}{2}$ and $1$ and accordingly the UUB bound of $\tilde{W}$ gets scaled down compared to that $\left(=\frac{b_N}{\lambda_{min}(M)}\right)$ with constant learning rate (also derived in Eq. (45) of  \cite{liu2015reinforcement}).
The UUB set for $\tilde{W}$ for constant learning rate gradient descent is $\|\tilde{W}\| > \frac{b_N}{\lambda_{min}(M)}$ for both Case (i) and (ii).
 This leads to critic NN weights converging close to their ideal weights in finite time.}
\end{remark}
\begin{remark}
\textnormal{Further, variable learning can scale the learning speed based on the instantaneous value of the HJB approximation error. 
This leads to faster convergence time as compared to constant learning rate gradient descents}
\end{remark}
These advantages are exemplified in the following section.
\vspace{-.1cm}
\section{Results and Simulation}\label{result}
\vspace{-.1cm}
In this section we will consider the numerical simulation of the parameter update law presented in this paper. 
\begin{enumerate}
    \item At first the parameter update law is validated on a generic nonlinear system with two different actuator limits and the results are contrasted against the constant learning parameter update law.
    \item Thereafter, the variable gain update law is validated on a full 6-DoF nonlinear model of UAV and the result is contrasted against the constant learning-based parameter update law.
\end{enumerate}
\vspace{-.2cm}
\subsection{Nonlinear system}
Consider a continuous time nonlinear system $\dot{x}=f(x)+g(x)u$ as mentioned in \cite{yang2015robust},
\begin{equation}
\begin{split}
f&=\begin{pmatrix}
-x_1+x_2 \\
-(x_1+1)x_2-49x_1+.5((\cos(x_1))^3\sin(x_2))
\end{pmatrix}=\begin{pmatrix}
f_1\\
f_2
\end{pmatrix} \\
g&=\begin{pmatrix}
0 \\
1
\end{pmatrix}=\begin{pmatrix}
g_1 \\
g_2
\end{pmatrix}
\end{split}
\label{sim_dyn}
\end{equation}
 Drift dynamics $f_1,f_2$ and control coupling dynamics ($g_1,g_2$) are as mentioned in (\ref{sim_dyn}).

This continuous time nonlinear system is required to track a desired reference system given as \cite{yang2015robust}.
\begin{equation}
\begin{split}
\begin{pmatrix}
\dot{x}_{d1} \\
\dot{x}_{d2}
\end{pmatrix}=
\begin{pmatrix}
x_{d2} \\
-49x_{d1}
\end{pmatrix}
\end{split}
\end{equation}
The augmented state vector $z=[e_{x1},e_{x2},x_{d1},x_{d2}]^T$. The Lyapunov function $L_2$ is selected as $L_2=1/2z^Tz$. Also, $R=1$, and $Q_1$ (refer to Eqs. (\ref{cost}), (\ref{ctr_cu})) is selected as,
\begin{equation}
Q_1=\begin{pmatrix}
I_2 & 0\\
0 & 0_{2\times 2}
\end{pmatrix}
\end{equation}
where, $I_2=diag(10,10)$.  
Regressor vector for critic network is selected as \cite{yang2015robust}. The larger the size of the regressor vector with multiple polynomial powers of augmented state $z$, the accurate will be the results \cite{hornik1990universal}.
\begin{equation}
\vartheta(z)=[z_1^2,z_2^2,z_3^2,z_4^2,z_1z_2,z_1z_3,z_1z_4,z_2z_3,z_2z_4,z_3z_4]^T
\end{equation}
Initial state of the system is chosen to be, $x(0)=[1.5,1.5]^T$. Critic weights are initialized to 0, i.e, $\hat{W}(0)=0$. 
A dithering noise of the form $n(t)=2e^{-0.009t}(\sin(11.9t)^2\cos(19.5t)+\sin(2.2t)^2\cos(5.8t)+\sin(1.2t)^2\cos(9.5t)+\sin(2.4t)^5)$ is added to maintain the persistent excitation (PE) condition \cite{vamvoudakis2014online}.
Now, a comparative study of the variable gain gradient descent method presented in this paper w.r.t. constant gradient descent will be carried out. 
In order to validate the performance of the controller developed in this paper, two input bounds were selected, i.e., $u_m=1.8$ and $u_m=9$. 
Figs. \ref{fig:1.8_with} and \ref{fig:1.8_wout} correspond to the case with input bound of 1.8, while Figs. \ref{fig:9_with} and \ref{fig:9_wout} correspond to the input bound of 9.
Constant learning rate ($\alpha$) for $u_m=9$ is selected to be 35.9 and discount ($\gamma=.1$), similarly for $u_m=1.8$, the constant learning rate $\alpha$ was 92.9 and discount being $\gamma=.1$. 
Constants used in variable gain gradient descent are $k_2=1.4$ for $u_m=9$ and $k_2=.7$ for $u_m=1.8$ (see Eq. (\ref{tuning_law1})). Simulations have been run till the critic NN weights have converged in respective cases.
\begin{figure*}
\centering
\subcaptionbox{Online training of Critic Weights\label{crit9}}{\includegraphics[width=.32\textwidth,height=9.5cm,keepaspectratio,trim={1.8cm 0.0cm 2cm .08cm},clip]{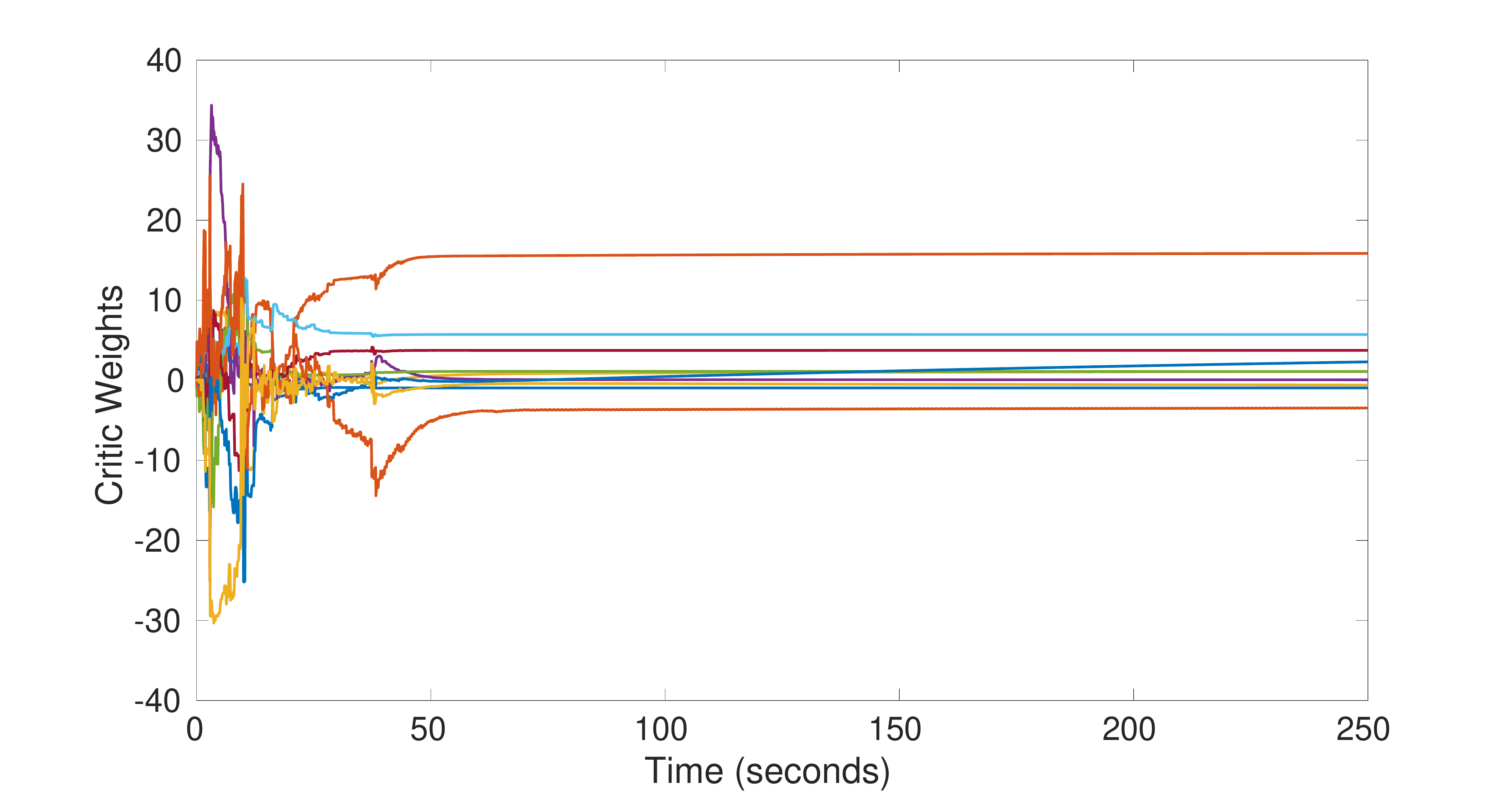}}%
\hspace{0cm} 
\subcaptionbox{State error\label{er9}}{\includegraphics[width=.32\textwidth,height=9.5cm,keepaspectratio,trim={1.8cm 0.0cm 2cm .08cm},clip]{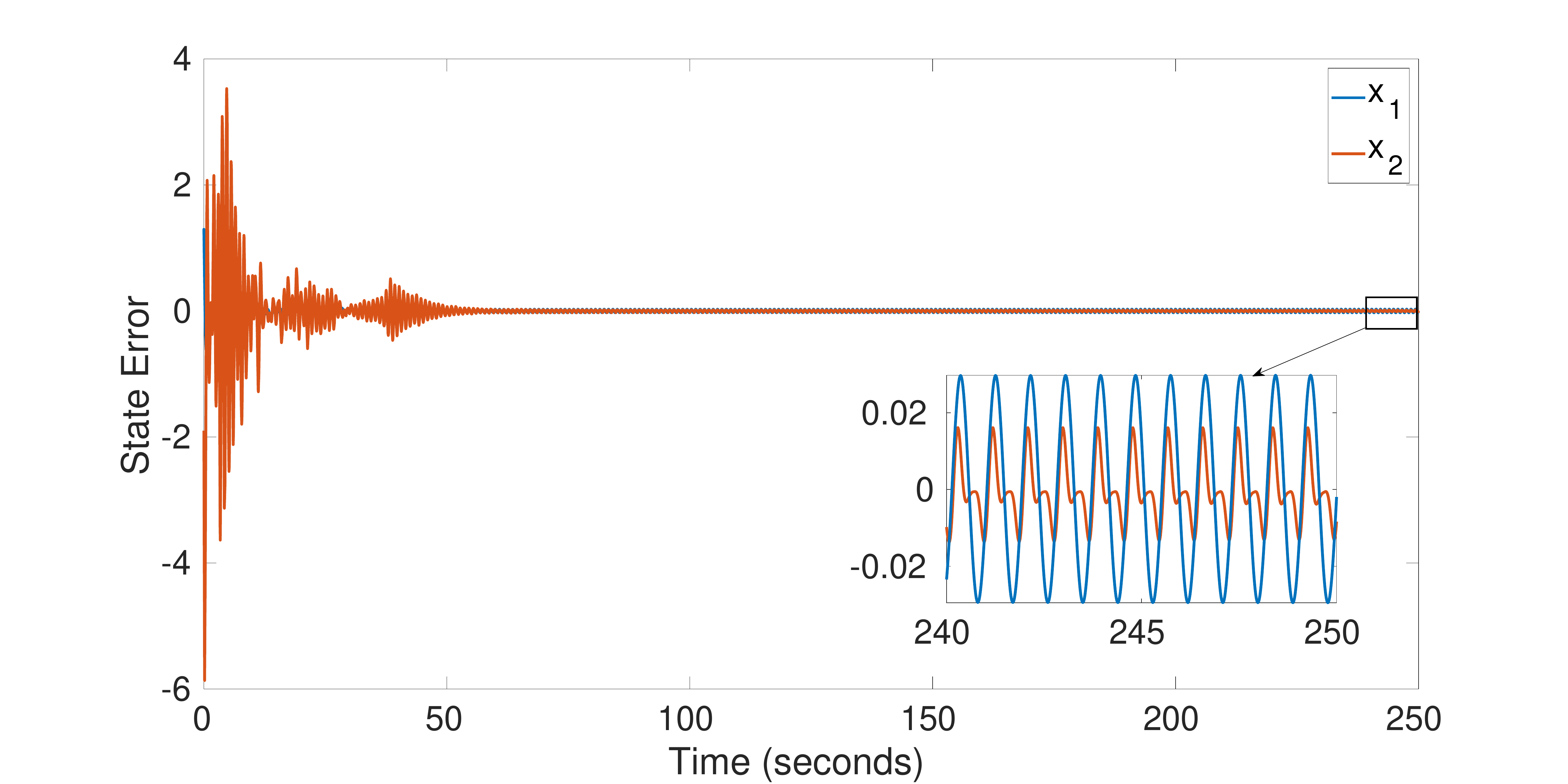}}%
\hspace{0cm} 
\subcaptionbox{Control action\label{ctr9}}{\includegraphics[width=.32\textwidth,height=9.5cm,keepaspectratio,trim={1.8cm 0.0cm 2cm .08cm},clip]{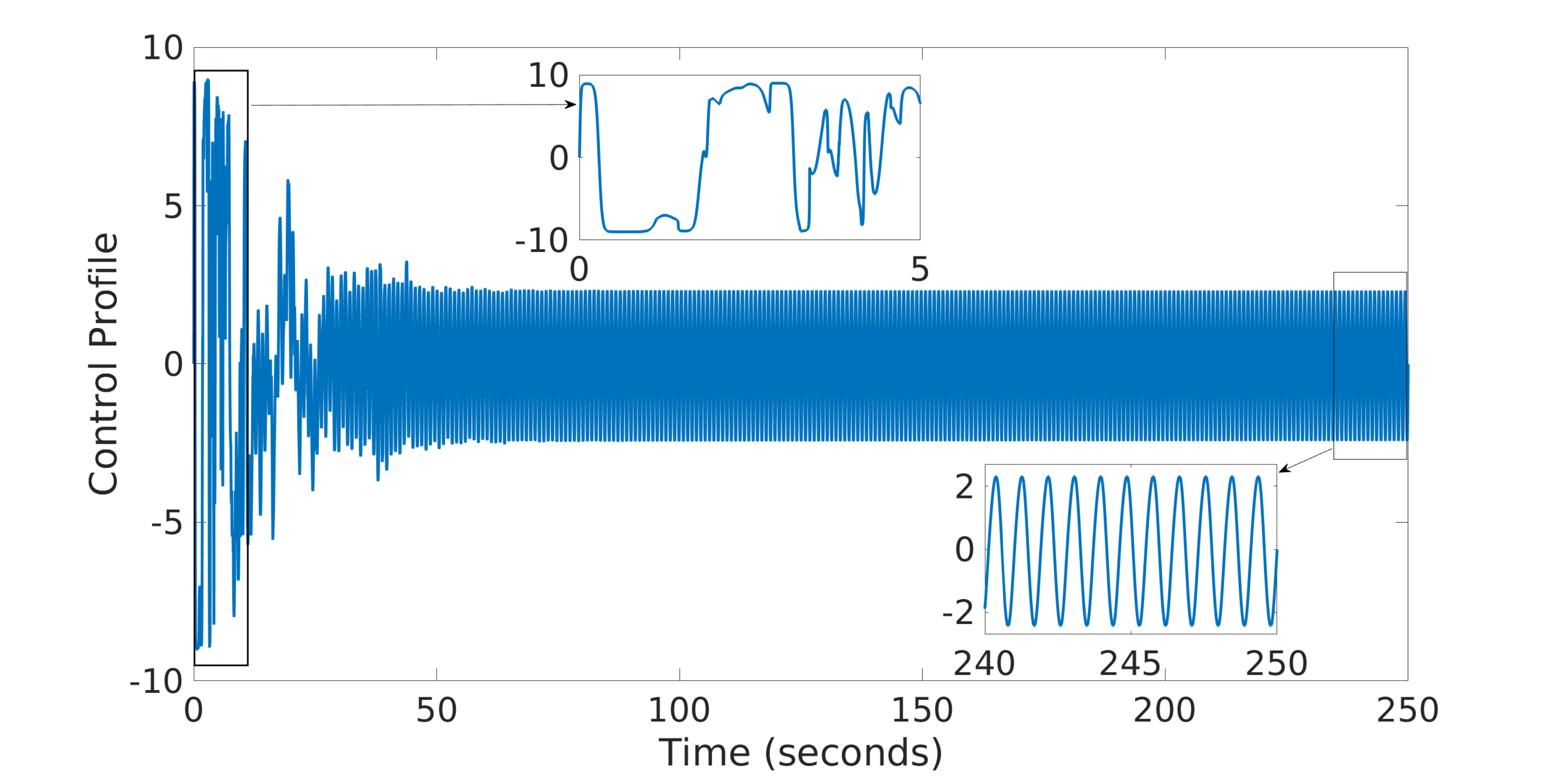}}%
\caption{Critic NN, state error and control profile with variable gain gradient descent (for $u_m=9$)}
\label{fig:9_with}
\end{figure*}
\begin{figure*}
\centering
\subcaptionbox{Online training of critic weights\label{crit9_wout}}{\includegraphics[width=.32\textwidth,height=9.5cm,keepaspectratio,trim={1.8cm 0.0cm 2cm .08cm},clip]{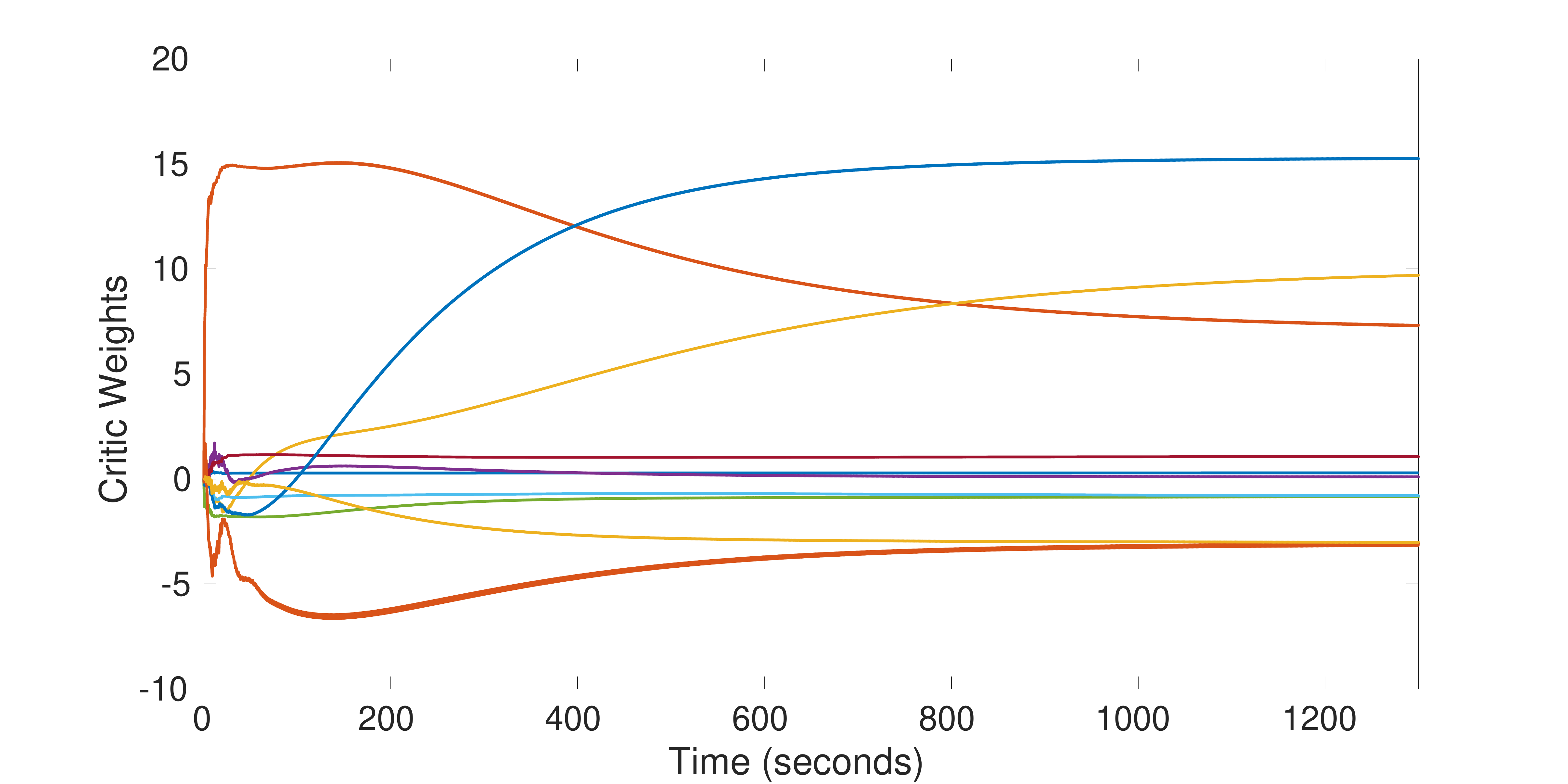}}%
\hspace{0cm} 
\subcaptionbox{State error\label{er9_wout}}{\includegraphics[width=.32\textwidth,height=9.5cm,keepaspectratio,trim={1.8cm 0.0cm 2cm .08cm},clip]{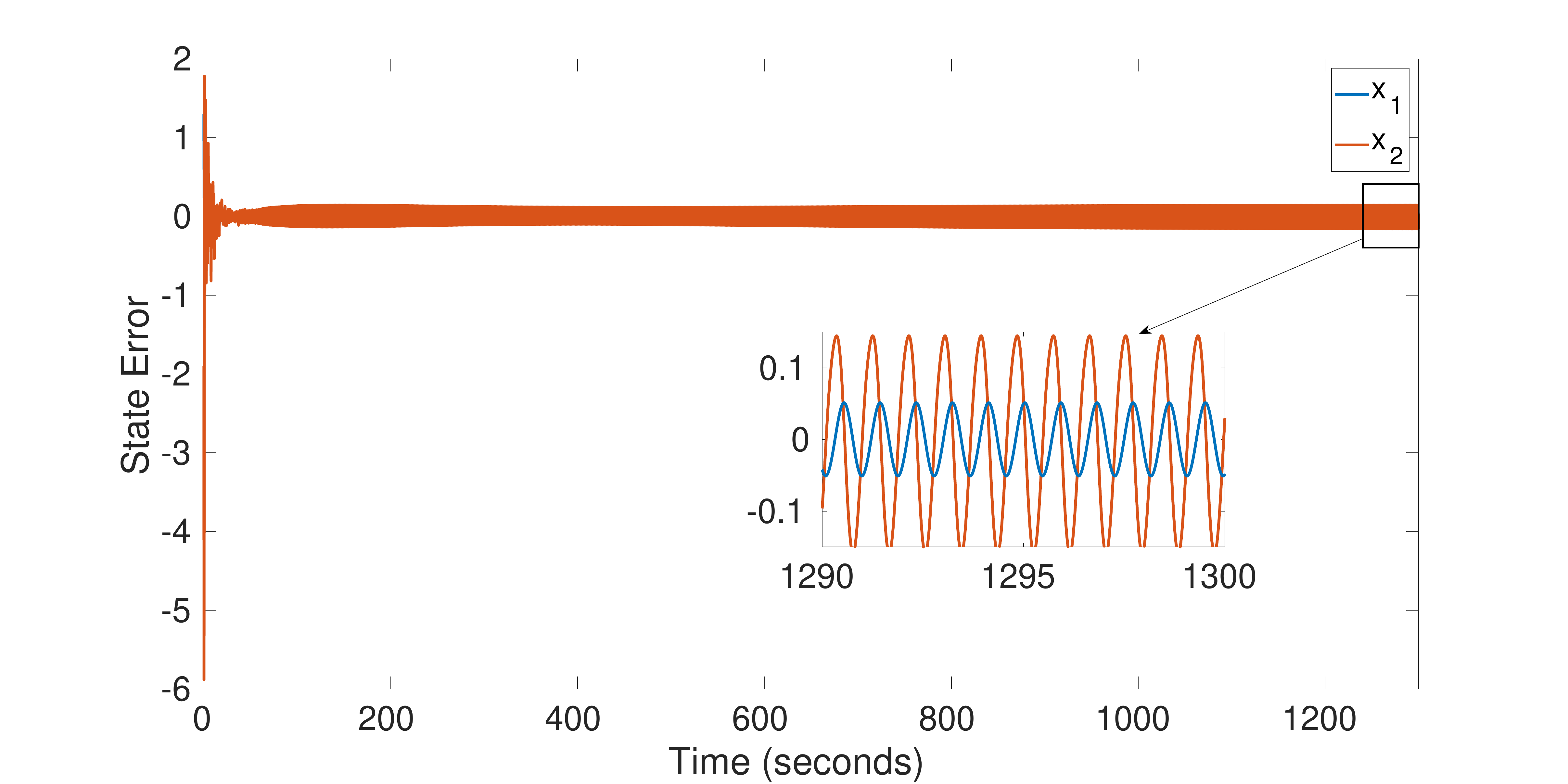}}%
\hspace{0cm} 
\subcaptionbox{Control action\label{ctr9_wout}}{\includegraphics[width=.32\textwidth,height=9.5cm,keepaspectratio,trim={1.8cm 0.0cm 2cm .08cm},clip]{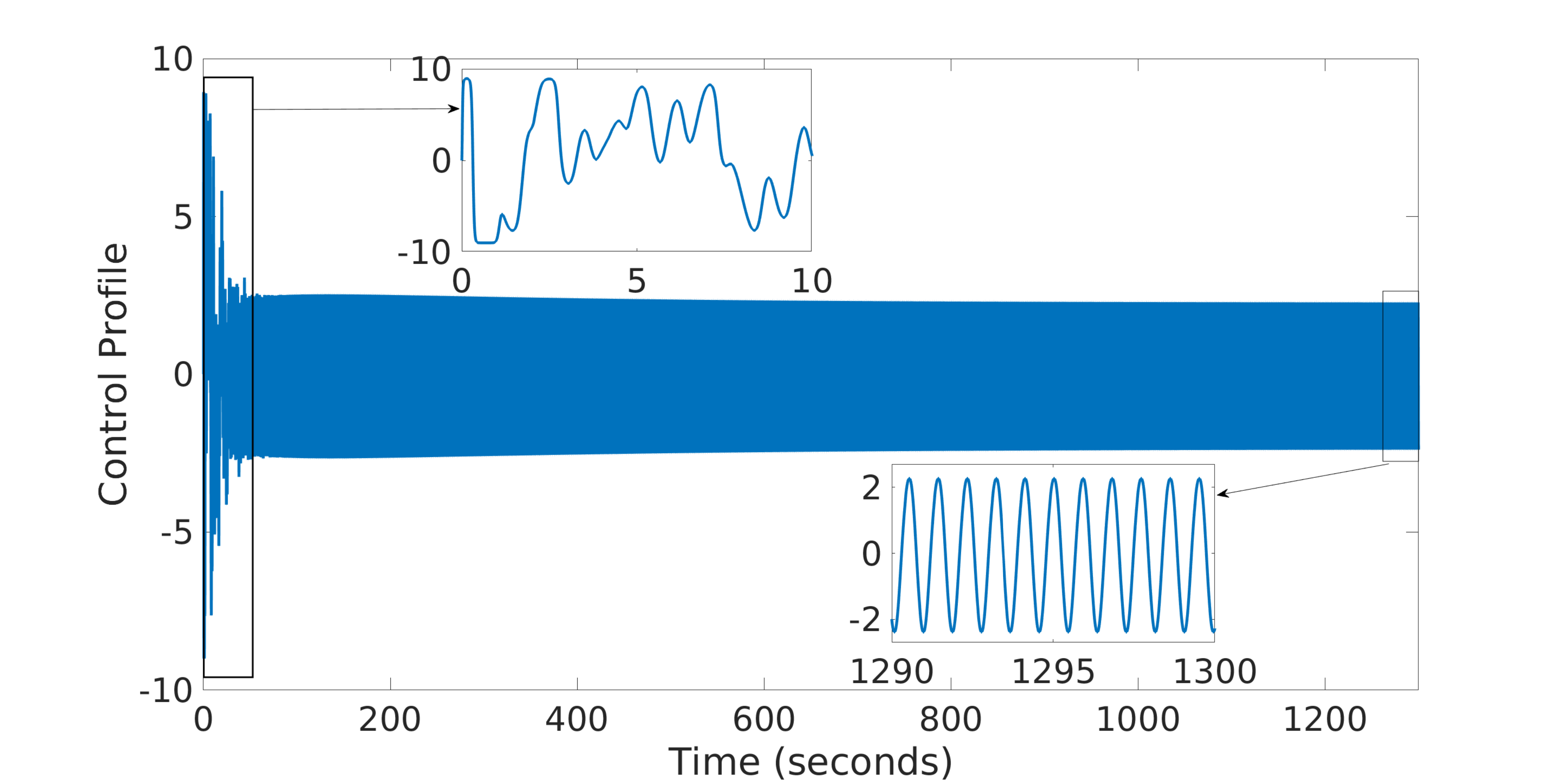}}%
\caption{Critic NN, state error and control profile with constant learning gradient descent (for $u_m=9$)}
\label{fig:9_wout}
\end{figure*}
\begin{figure*}
\centering
\subcaptionbox{Online training of critic weights\label{crit1.8_with}}{\includegraphics[width=.32\textwidth,height=9.5cm,keepaspectratio,trim={1.8cm 0.0cm 2cm .08cm},clip]{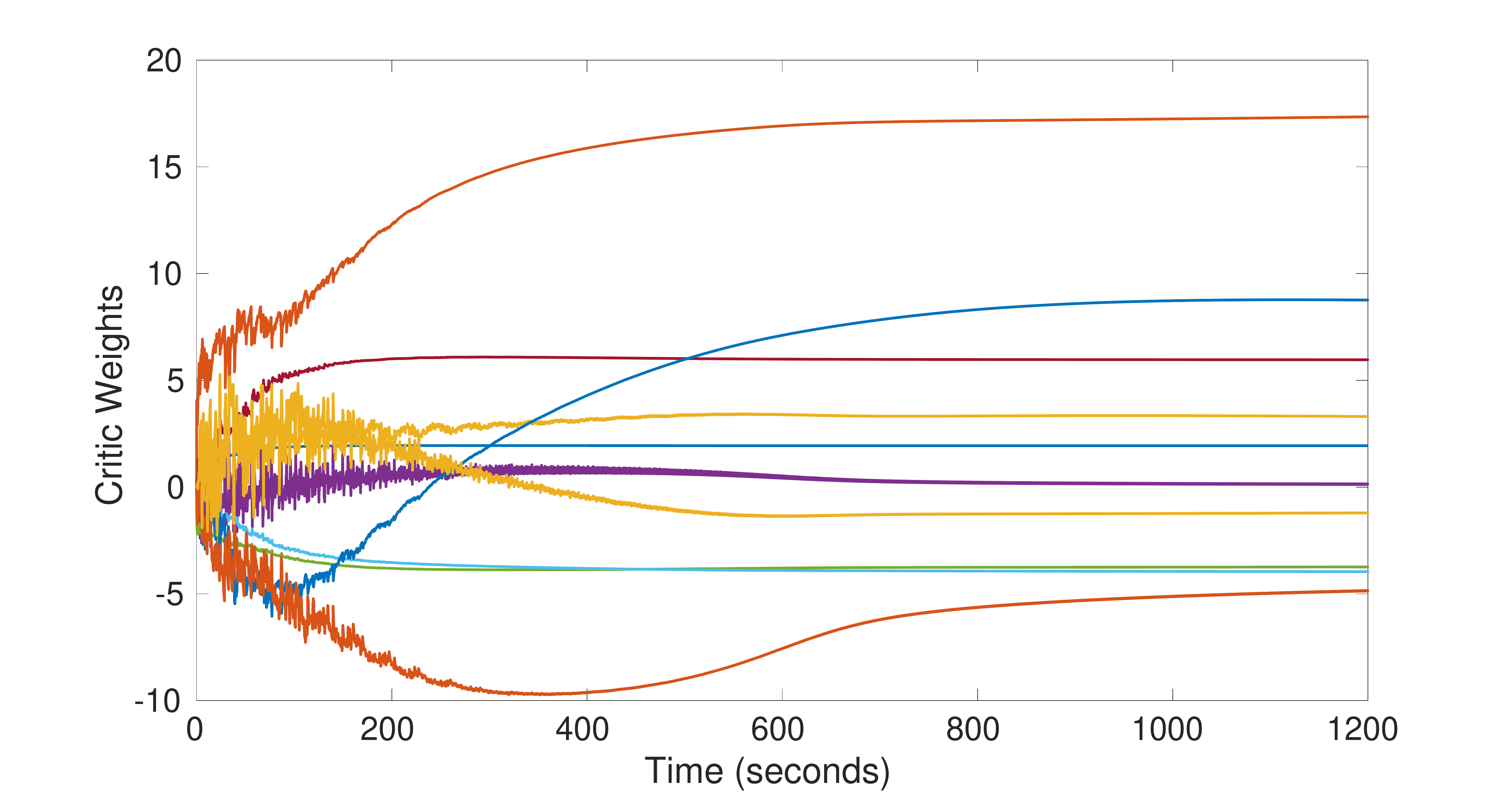}}%
\hspace{0cm} 
\subcaptionbox{State error\label{er1.8_with}}{\includegraphics[width=.32\textwidth,height=9.5cm,keepaspectratio,trim={1.8cm 0.0cm 2cm .08cm},clip]{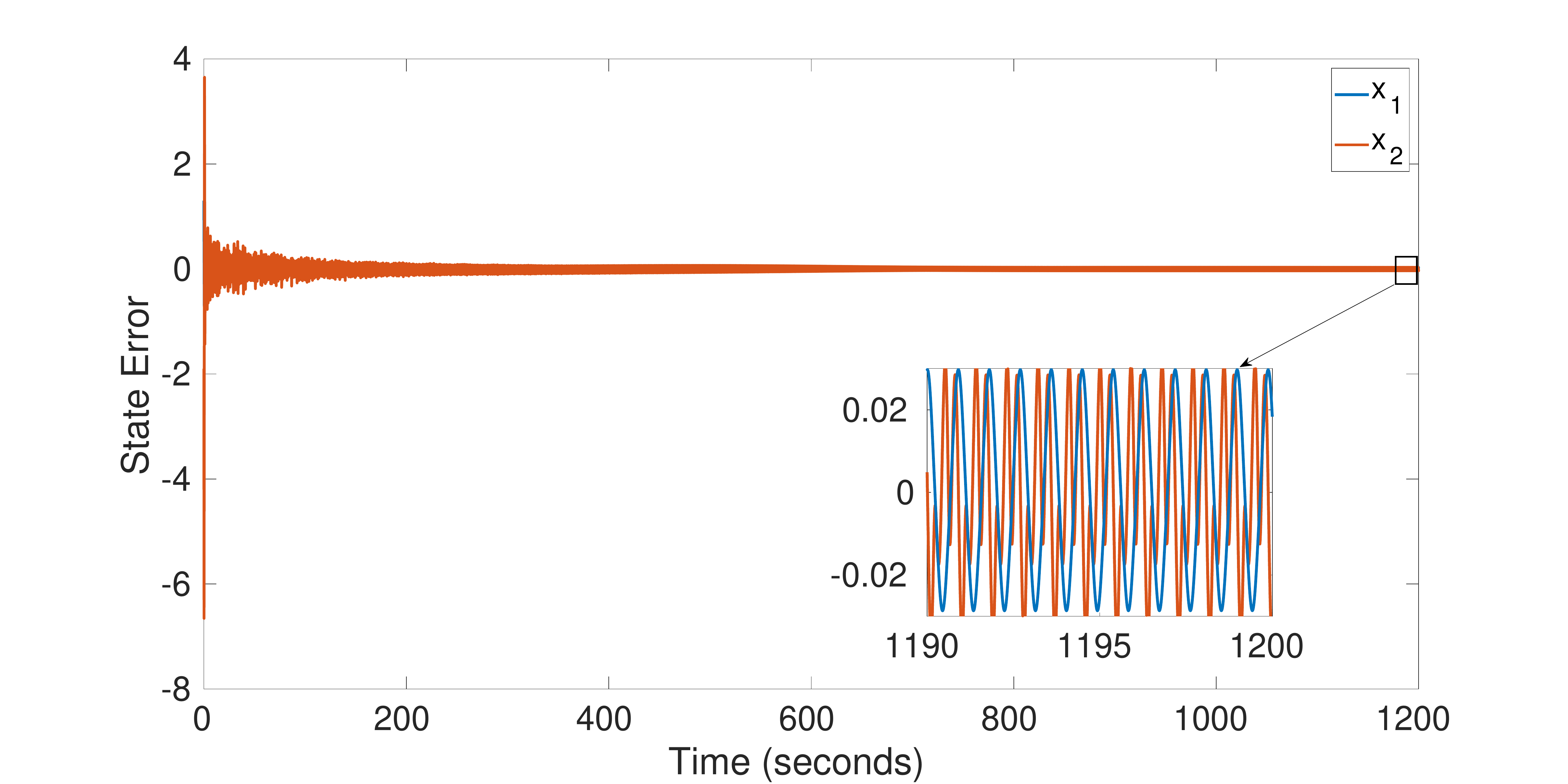}}%
\hspace{0cm} 
\subcaptionbox{Control action\label{ctr1.8_with}}{\includegraphics[width=.32\textwidth,height=9.5cm,keepaspectratio,trim={1.8cm 0.0cm 2cm .08cm},clip]{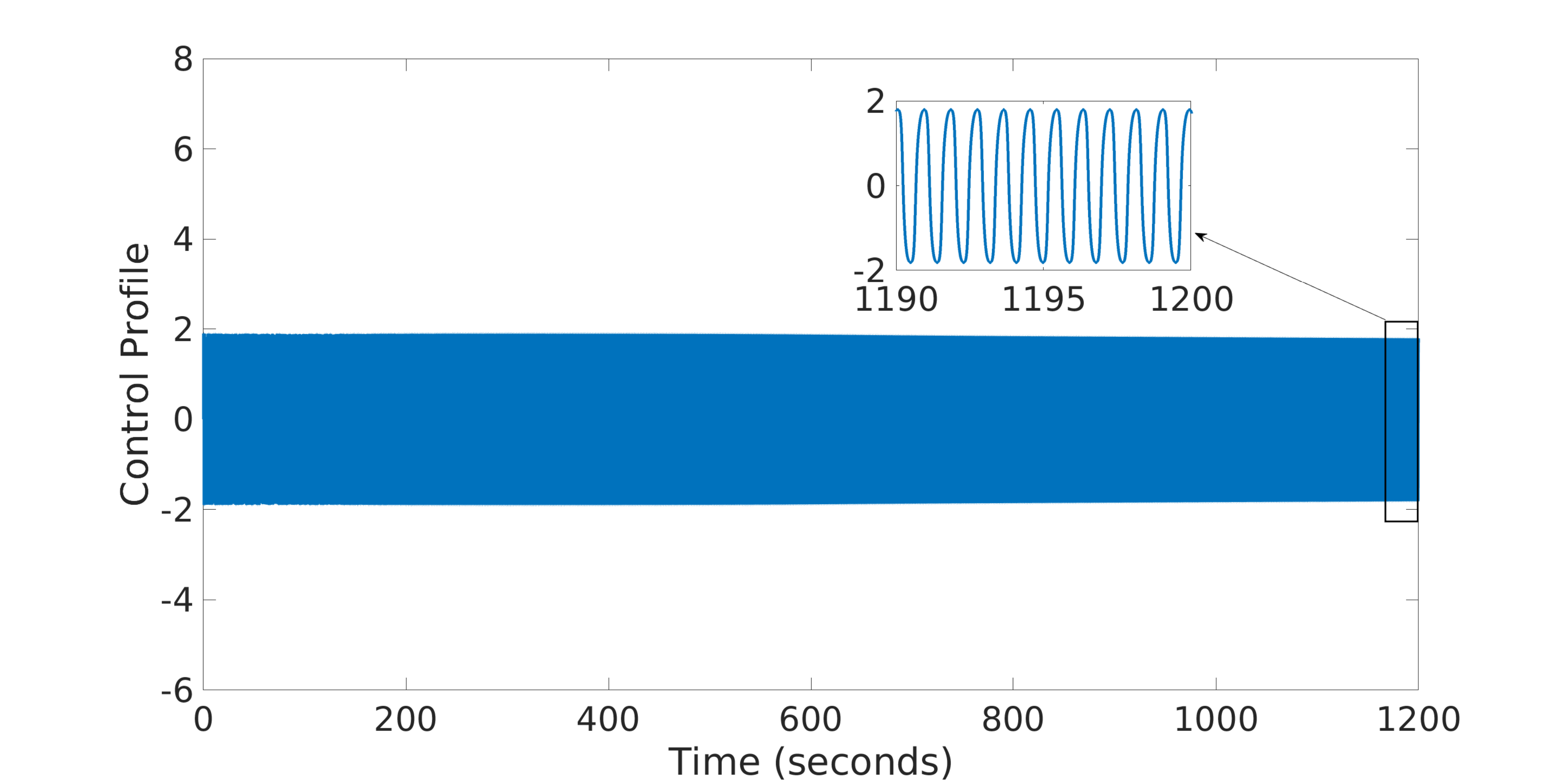}}%
\caption{Critic NN, state error and control profile with variable gain gradient descent (for $u_m=1.8$)}
\label{fig:1.8_with}
\end{figure*}

\begin{figure*}
\centering
\subcaptionbox{Online training of critic weights\label{crit1.8_wout}}{\includegraphics[width=.32\textwidth,height=9.5cm,keepaspectratio,trim={1.8cm 0.0cm 2cm .08cm},clip]{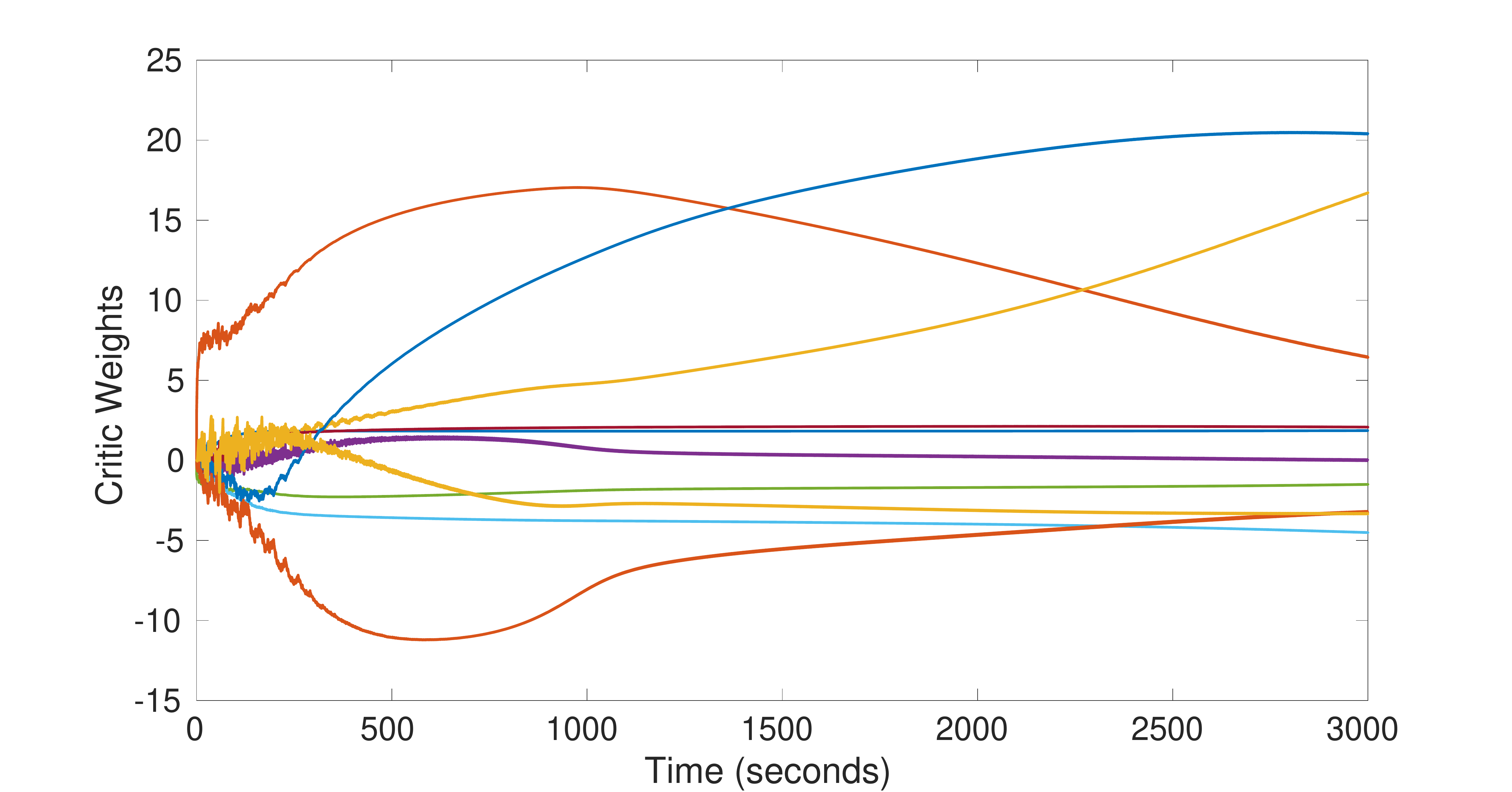}}%
\hspace{0cm} 
\subcaptionbox{State error\label{er1.8_wout}}{\includegraphics[width=.32\textwidth,height=9.5cm,keepaspectratio,trim={1.8cm 0.0cm 2cm .08cm},clip]{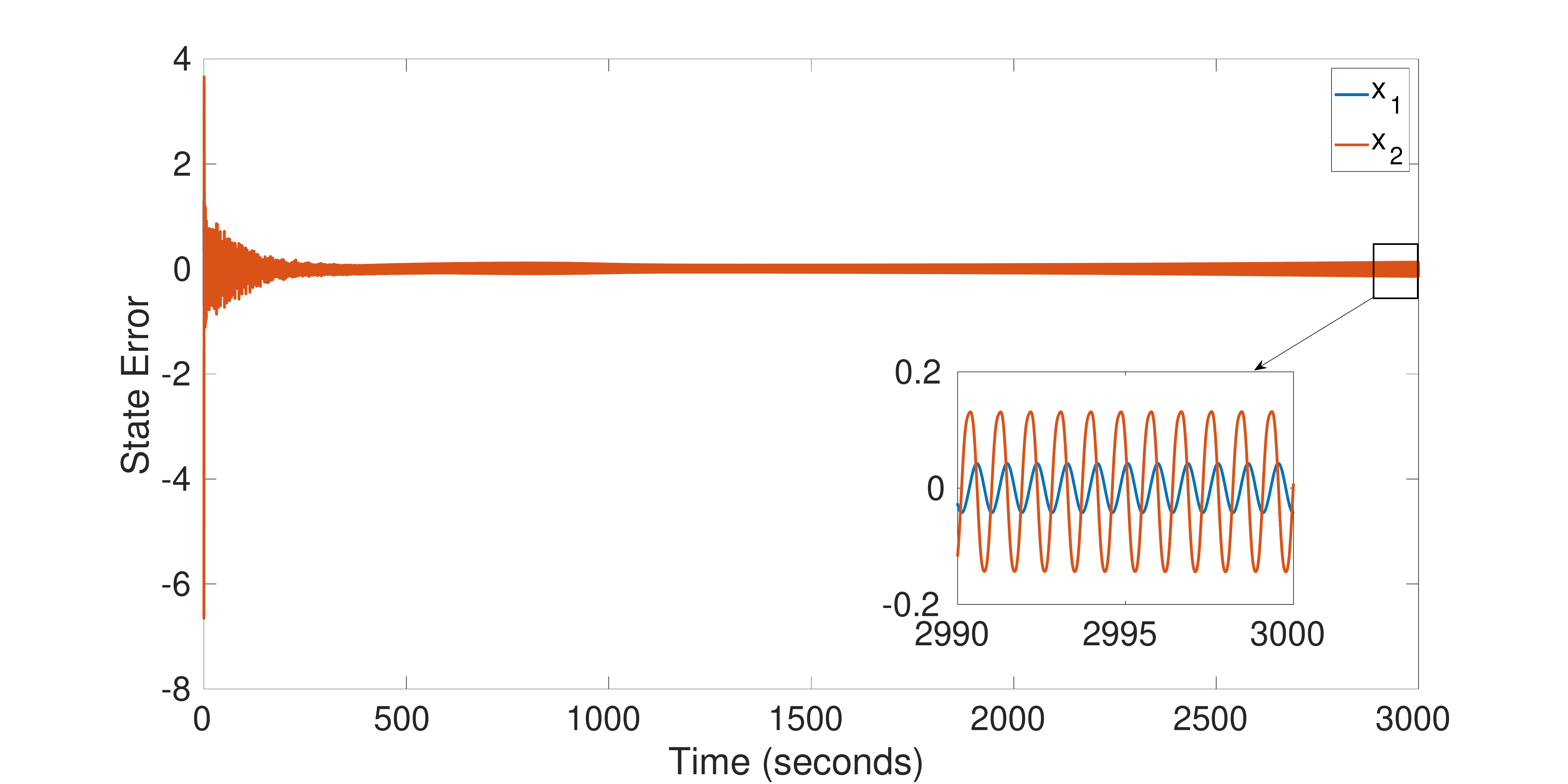}}%
\hspace{0cm} 
\subcaptionbox{Control action\label{ctr1.8_wout}}{\includegraphics[width=.32\textwidth,height=9.5cm,keepaspectratio,trim={1.8cm 0.0cm 2cm .08cm},clip]{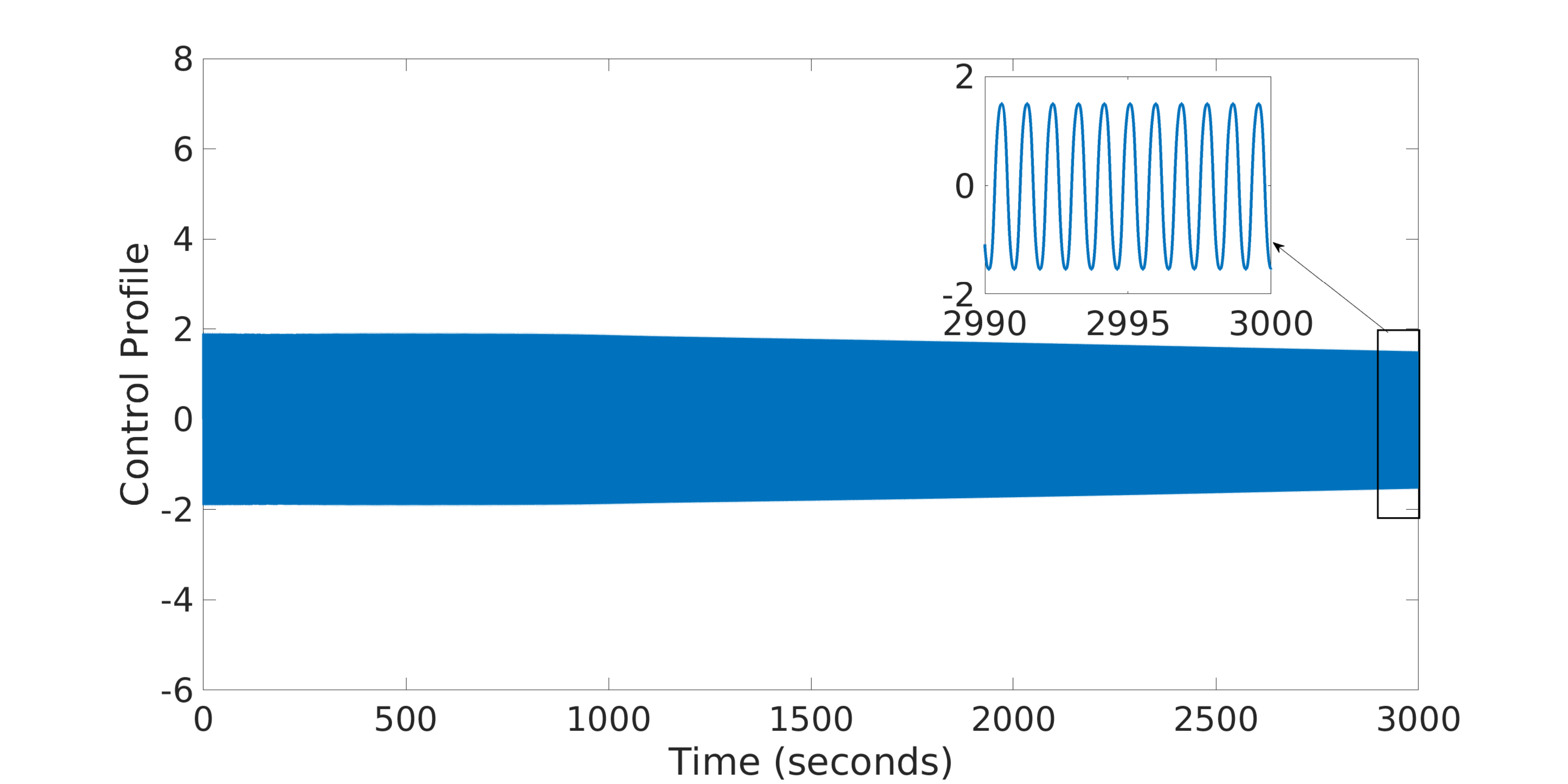}}%
\caption{Critic NN, State error and Control Profiles with constant learning gradient descent (for $u_m=1.8$)}
\label{fig:1.8_wout}
\end{figure*}
Comparing Figs. \ref{crit9} and \ref{crit9_wout}, it can be observed that the application of variable gain gradient descent leads to faster and efficient learning of critic NN weights. 
In Fig. \ref{crit9} when the variable gain gradient descent algorithm was utilized, the critic NN weights converged much before 250s, whereas in Fig. \ref{crit9_wout}, they took approximately 1300s.
Observe that the update law presented in this paper is able to bring the state error to a much tighter residual set compared to constant learning rate-based gradient descents as can be seen from Figs. \ref{er9} and \ref{er9_wout}.
This is due to the fact that, the evolution of weight vectors with variable gain gradient descent converges to a smaller neighborhood about the ideal weight vector than with constant learning rate gradient descent.

It is noted from Figs. \ref{ctr9} and \ref{ctr9_wout} that the optimal control commands generated were within the saturation limit of $[-9,9]$.
Also, most of the time, the control effort was well within this bounded interval $[-2.2,2.2]$. 
Hence, in order to study the performance of the proposed adaptation scheme in a more stringent setting, a tighter control saturation limit $u_m = 1.8<2.2$ would be considered next. 

A tighter input saturation has an adverse effect on learning as it takes more time to achieve convergence of critic NN weights as can be seen from Figs.  \ref{crit1.8_with} and \ref{crit1.8_wout}.
In contrast to constant learning gradient descents, variable gain gradient descent-based update law presented in this paper is able to not only achieve convergence of critic NN weights within 1200s (refer to Fig. \ref{crit1.8_with}) but also bring the state error (refer to Fig. \ref{er1.8_with}) to a tight residual set comparable to Fig \ref{er9}.
On the other hand, under constant learning-based update law, some of the critic NN weights were not able to converge properly even after 3000s (refer to Fig. \ref{crit1.8_wout}).
This results in larger state error as can be seen in Fig. \ref{er1.8_wout}.

It is because of these reasons, that the update law presented in this paper yields improved tracking performance even with tight actuator constraints.
The control effort was limited to $[-1.8,1.8]$ for both with/without variable gain gradient descent update law as can be seen in Figs. \ref{ctr1.8_with} and \ref{ctr1.8_wout}.




In \cite{yang2015robust}, a smaller value of learning rate was selected to train the critic network online, however, following their formulation, it takes a lot more time for the controller to bring the oscillation magnitude of state error down to a small bound, which can be clearly seen in Fig. 1 in \cite{yang2015robust}.
Also in \cite{yang2015robust}, the control effort during the initial phases touches $[-20,20]$ and does not incorporate actuator constraints. 
As it can be inferred from Figs. \ref{er9_wout} and \ref{er1.8_wout} that a high constant learning rate leads to larger oscillation bound on states error compared to the case when variable gain gradient descent (see Figs. \ref{er9} and \ref{er1.8_with}) was utilized. 
It can be clearly concluded from Figs. \ref{fig:9_with} and \ref{fig:1.8_with} that, the variable gain gradient descent based tuning law proposed in this paper yields faster learning and is able to successfully bring the state error to a much tighter residual set than constant learning rate for both actuator contraints limits considered in this paper, i.e., $u_m=9$ and $u_m=1.8$. 

Thus, the prime advantage of variable gain gradient descent-based critic update law is the ability to select reasonably high learning rates without large steady state errors.  
\subsection{Full 6-DoF nonlinear model of UAV}
\begin{figure*}
\centering
\subcaptionbox{Desired and actual attitude profile\label{fig:att_w}}{\includegraphics[width=.32\textwidth,height=9.5cm,keepaspectratio,trim={1.8cm 0.0cm 4cm .08cm},clip]{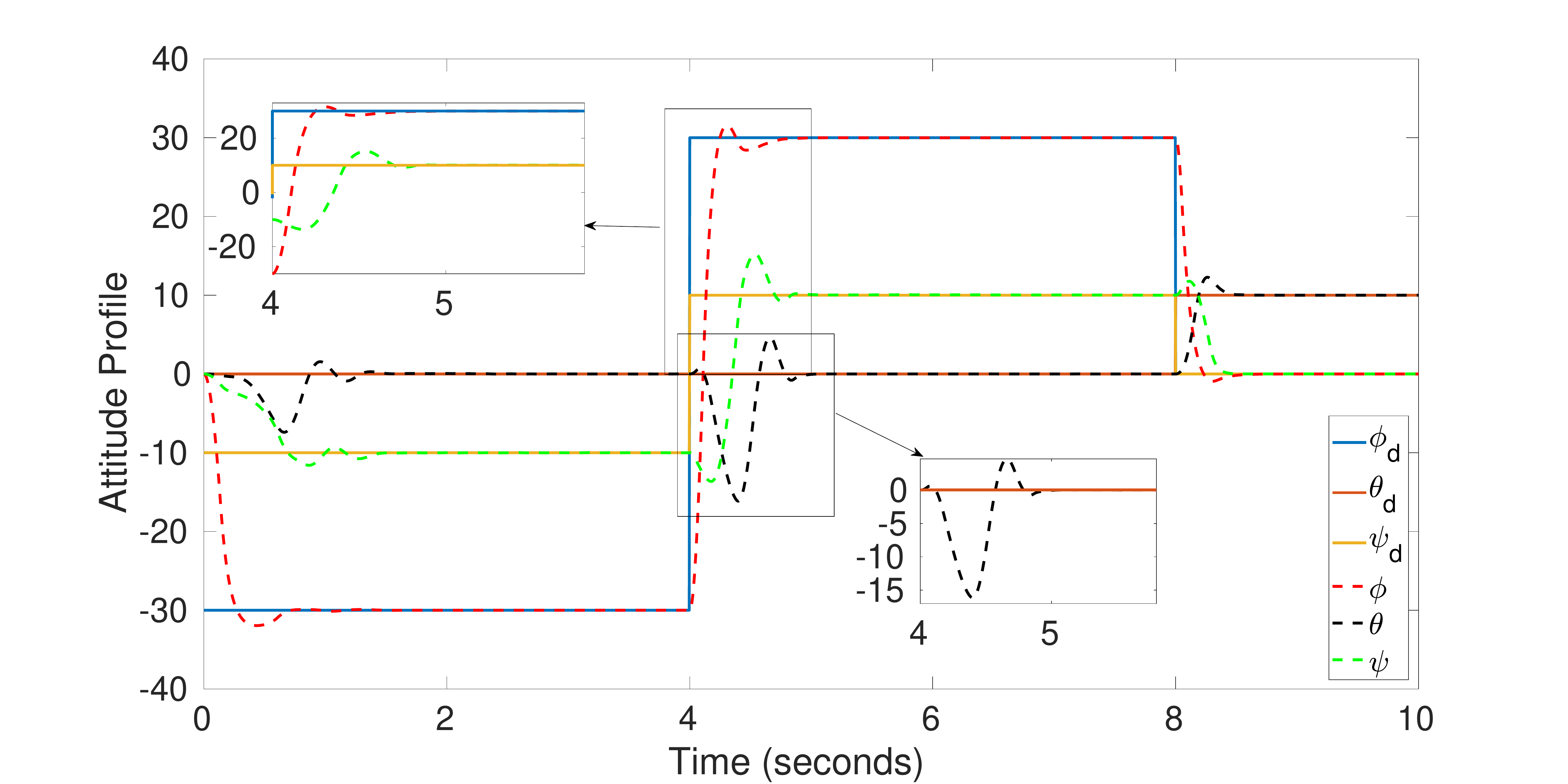}}%
\hspace{0cm} 
\subcaptionbox{Control profile representing elevator, eileron and rudder\label{fig:ctr_w}}{\includegraphics[width=.32\textwidth,height=9.5cm,keepaspectratio,trim={1.8cm 0.0cm 2cm .08cm},clip]{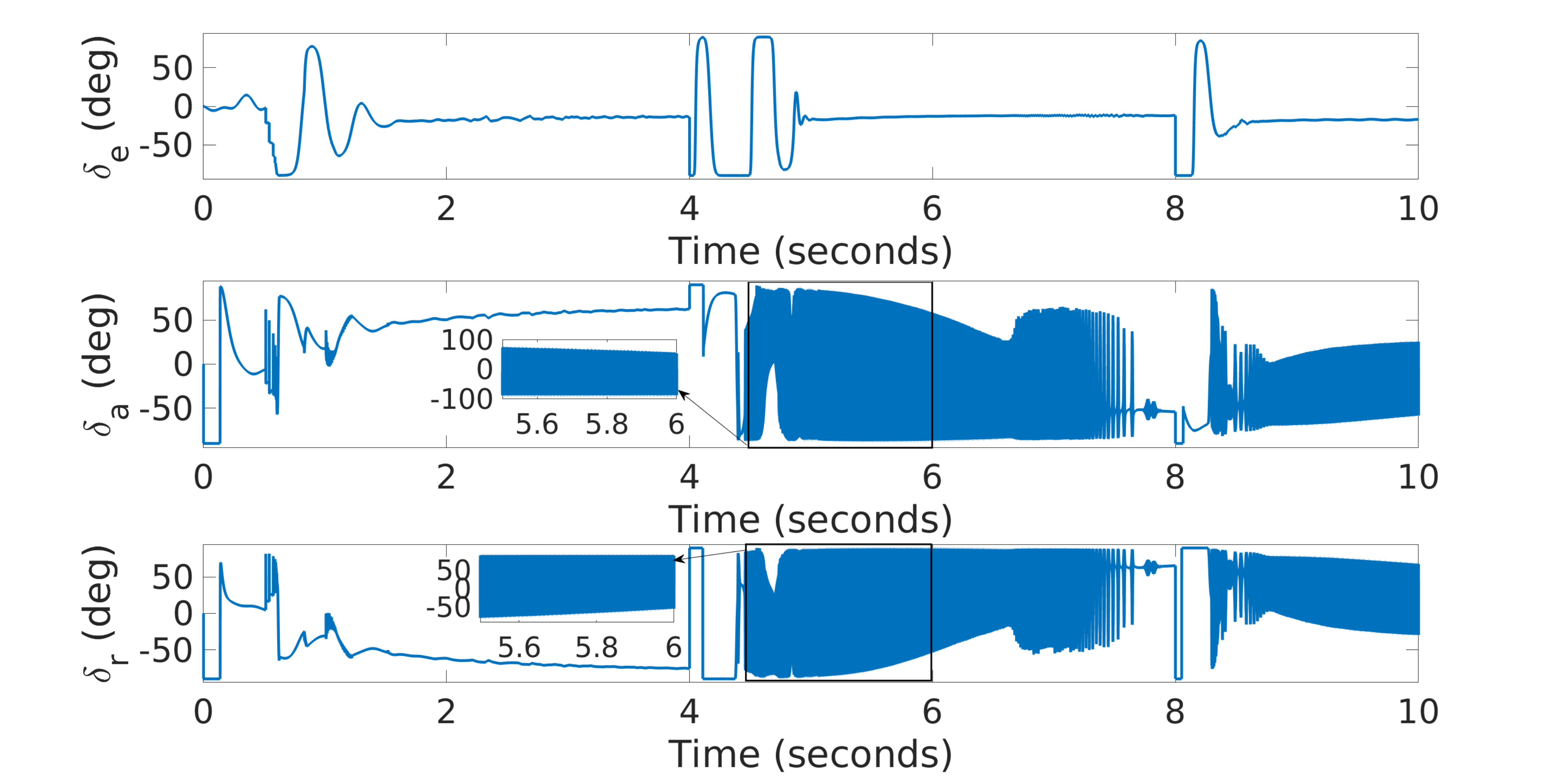}}%
\hspace{0cm}
\subcaptionbox{Online training of critic NN weights\label{fig:crit_w}}{\includegraphics[width=.32\textwidth,height=9.5cm,keepaspectratio,trim={1.8cm 0.0cm 2cm .08cm},clip]{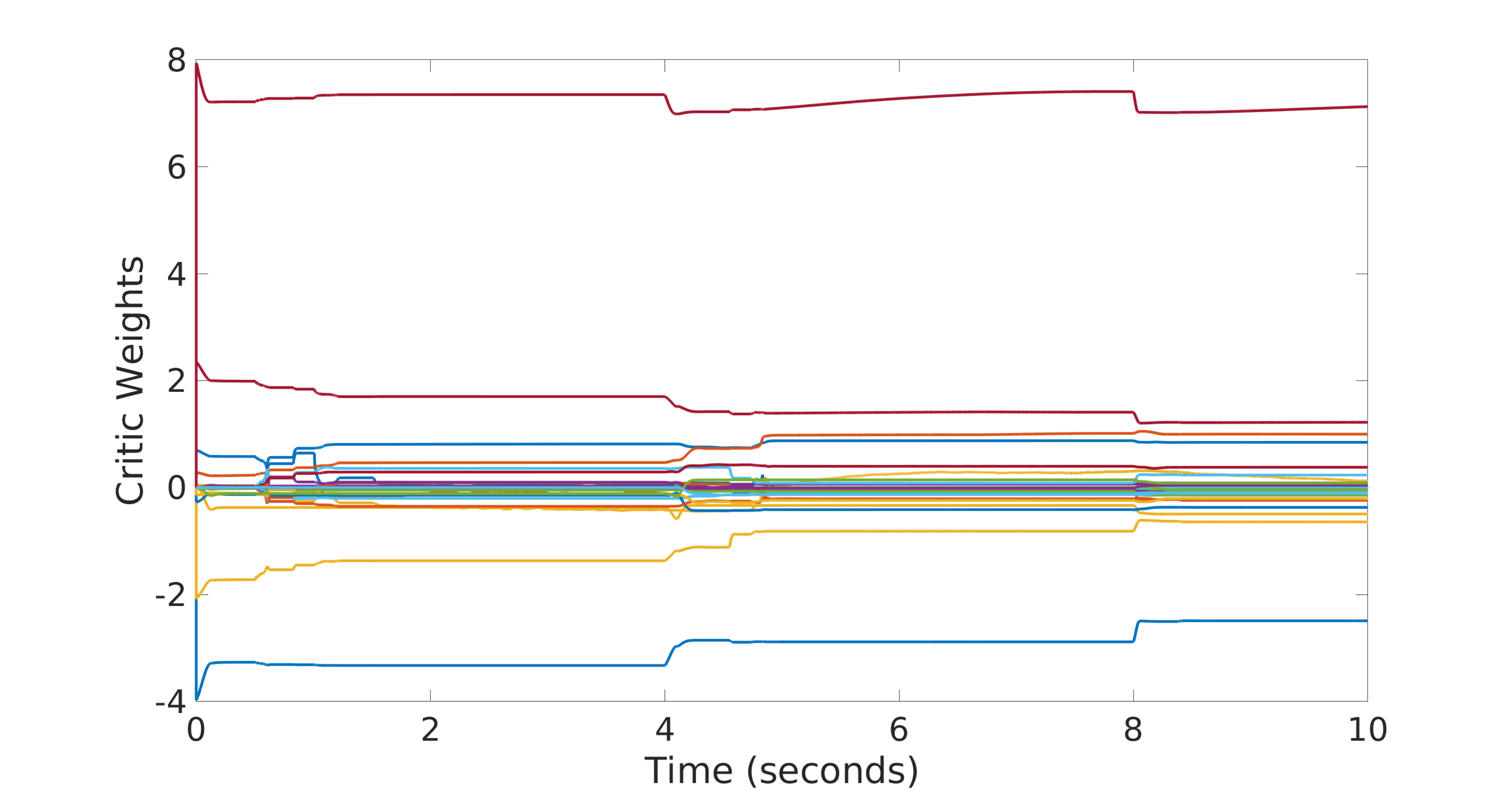}}
%
\caption{Performance of Aerosonde UAV under constant learning-based update law}
\label{fig:e_states1}
\end{figure*}
In this subsection, the variable gain gradient descent update law presented in this paper is validated on the full 6-DoF nonlinear model of the Aerosonde UAV (refer to Pages 61 and 276 of \cite{beard2012small}), and its performance is compared against that of the constant learning-based update law. The requirement of the control scheme is to track the desired attitude angles of the UAV. Desired set points for $\phi,\theta,\psi$ (roll, pitch and yaw angle, respectively) were set to $(-30,0,-10)$ degrees for first $4$ seconds, then $(30,0,10)$ degrees from 4-8 seconds and finally $(0,0,10)$ degrees.

The control implementation is made up of two cascaded loops, the the first loop, i.e., outer loop converts the desired Euler angle information to desired rates, the inner loop uses the developed control algorithm to track the desired rates in an optimal way. Desired Euler angle rates are given by, $p_{des}=\dot{\phi}_{des}-8 e_{\phi},~q_{des}=\dot{\theta}_{des}-10 e_{\theta},~r_{des}=\dot{\psi}_{des}-12 e_{\psi}$, where $\phi,\theta,\psi$ are roll, pitch and yaw angles, respectively. 
The deflection of elevator, aileron and rudder forms the control input to the UAV (represented by $\delta_e,\delta_a,\delta_r$ respectively). The control deflections are limited to $\pm 90$ degrees.

The augmented state is, $z=[e_{p},e_{q},e_{r},p_{des},q_{des},r_{des}]^T \in \mathbb{R}^6$ where $e=x-x_{des}$ and $x=[p,q,r]^T$. 
The regressor vector for critic NN is chosen to be, $\vartheta=[z_1,z_2,z_3,z_4,z_5,z_6,z_1^2,z_2^2,z_3^2,z_4^2\\
,z_5^2,z_6^2,z_1z_2,z_1z_3,z_1z_4,z_1z_5,z_1z_6,z_2z_3,z_2z_4,z_2z_5,z_2z_6,z_3z_4\\
,z_3z_5,z_3z_6,z_4z_5,z_4z_6,z_5z_6]^T$.
The discount factor was selected as $\gamma=0.1$, the weight matrix for augmented states and control are  $Q_1=diag(10,10,50,0,0,0)$ and $R=I_3$, respectively. 
The baseline learning rate $\alpha=14.7$, parameters for variable gain gradient descent are 
$k_2=.1$.
A dithering noise of the form, $n(t)=2e^{-0.009t}(\sin(11.9t)^2\cos(19.5t)+\sin(2.2t)^2\cos(5.8t)+\sin(1.2t)^2\cos(9.5t)+\sin(2.4t)^5)$ is added to maintain the persistent excitation (PE) condition as demonstrated in \cite{vamvoudakis2014online}. All the critic weights were initialized to $0$, i.e., $\hat{W}(0)=0$.
\begin{figure*}
\centering
\subcaptionbox{Desired and actual attitude profile\label{fig:att}}{\includegraphics[width=.32\textwidth,height=9.5cm,keepaspectratio,trim={1.8cm 0.0cm 4cm .08cm},clip]{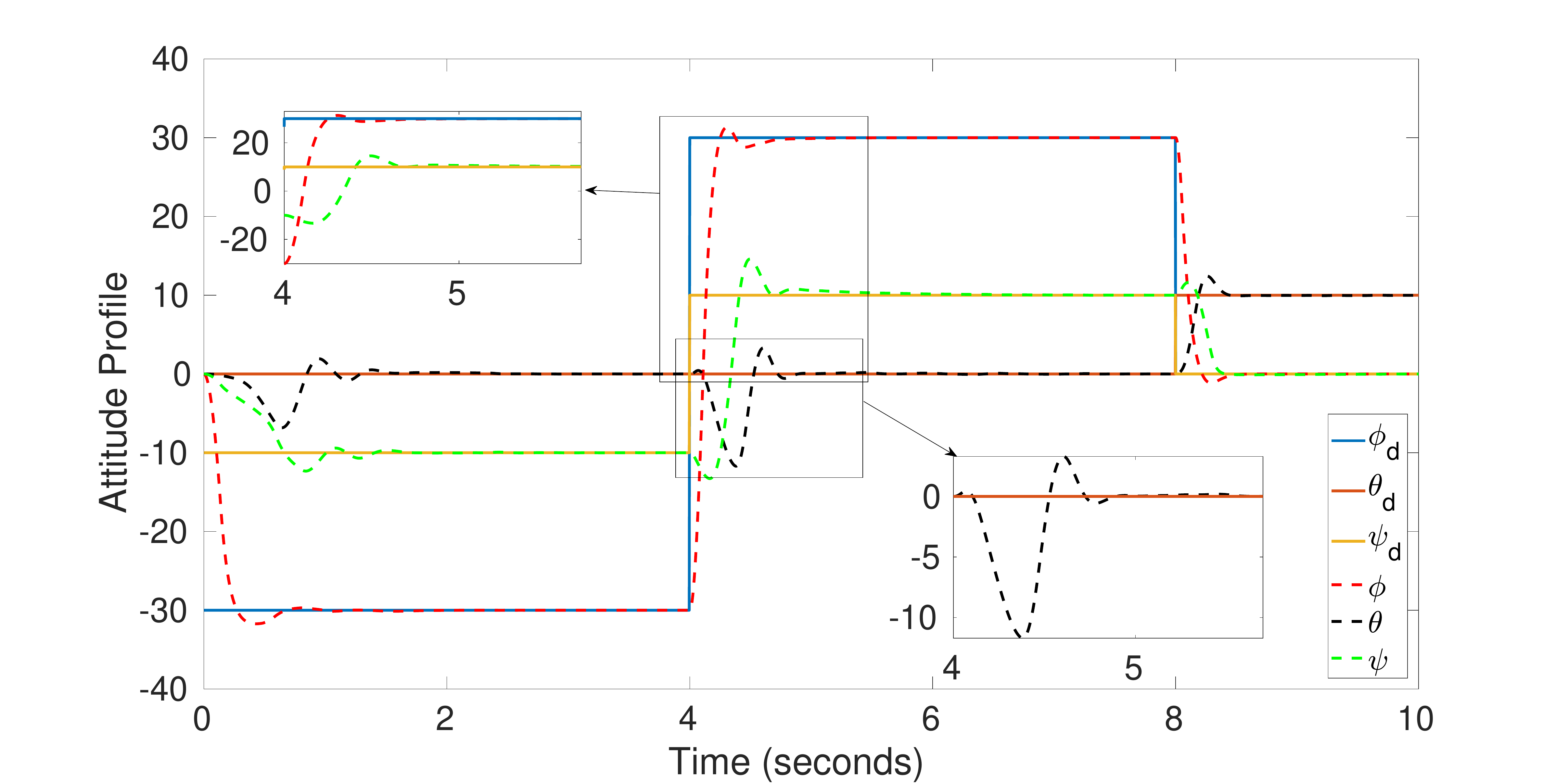}}%
\hspace{0cm} 
\subcaptionbox{Control profile representing elevator, eileron and rudder\label{fig:ctr}}{\includegraphics[width=.32\textwidth,height=9.5cm,keepaspectratio,trim={1.8cm 0.0cm 2cm .08cm},clip]{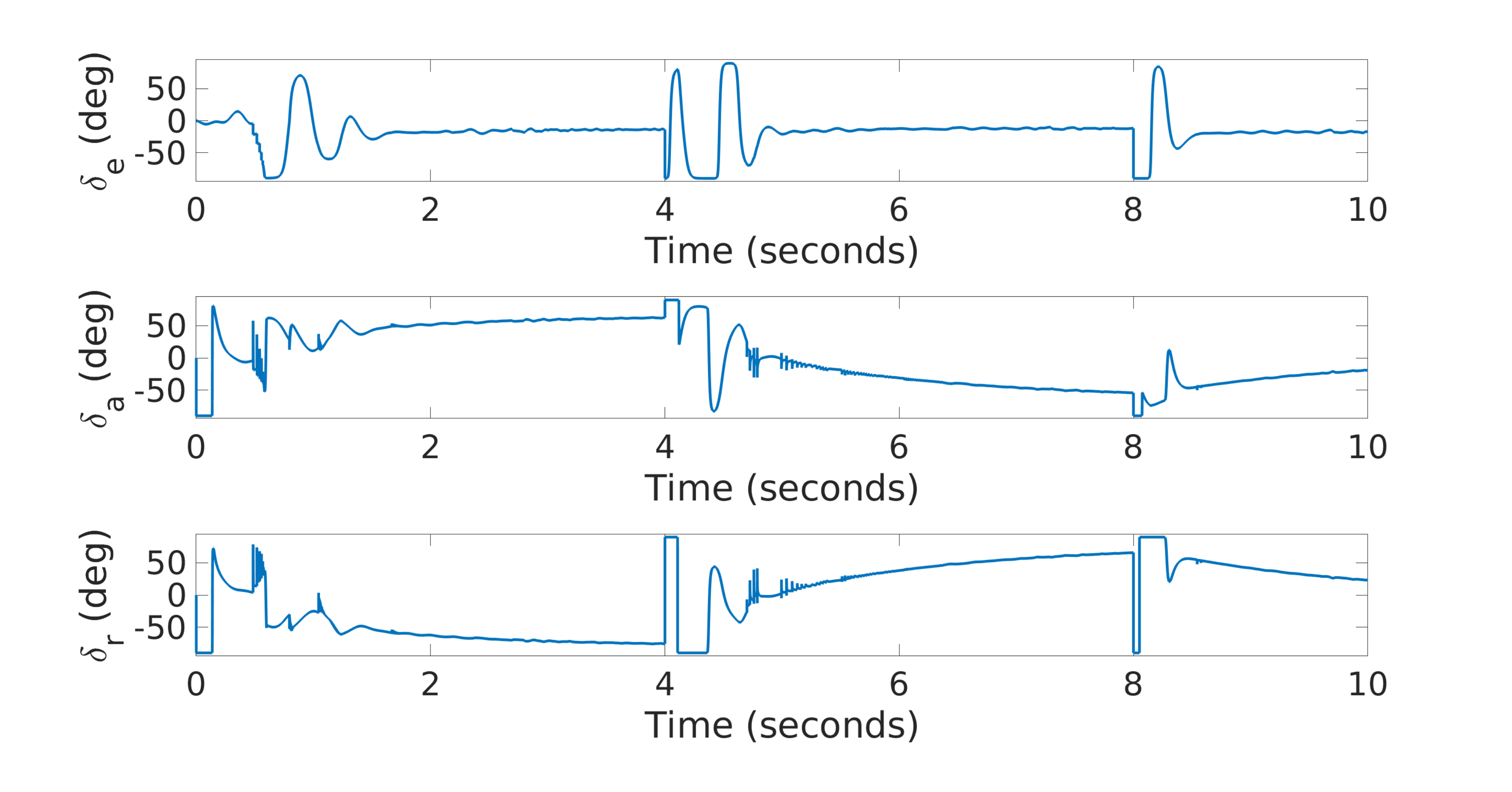}}%
\hspace{0cm}
\subcaptionbox{Online training of critic NN weights\label{fig:crit}}{\includegraphics[width=.32\textwidth,height=9.5cm,keepaspectratio,trim={1.8cm 0.0cm 2cm .08cm},clip]{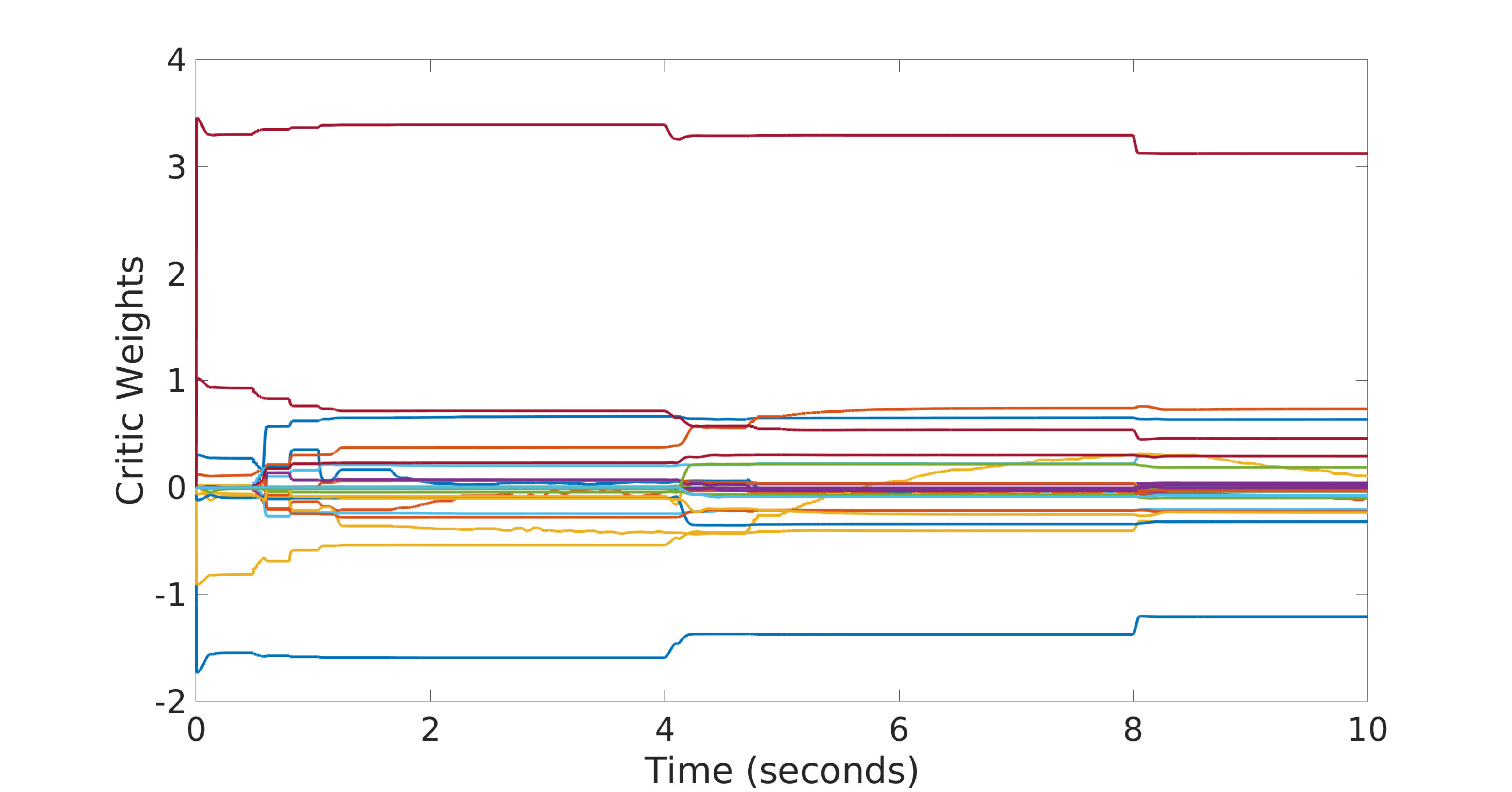}}
%
\caption{Performance of Aerosonde UAV under variable gain gradient descent-based update law}
\label{fig:e_states}
\end{figure*}
Both the update laws, i.e., variable gain gradient descent and constant learning rate update law were run with same set of parameters except the exponents in variable gain term and are able to track the desired reference set point.
However, it can be seen from Figs. \ref{fig:crit} and \ref{fig:crit_w} that the critic weights undergo spike at the times when reference command for attitude changes i.e., at $4$ and $8$ second. 
At these junctures it can be clearly noticed that the critic weights when updated via the variable gain gradient descent-update law converge properly within sufficiently short time-span and before the next attitude reference signal changes, i.e., during the intervals, $0-4$ sec, $4-8$ sec and then finally $8-10$ sec.
However, on the other hand, the critic weights are not able to converge properly in such short time-span between changes in the reference attitude when updated via constant learning-based update law.
As critic weights have converged close to their ideal values in very small time when updated via the variable gain gradient descent update law, overshoots in state errors in this case are smaller in comparison with that in case of constant gain gradient descent, as can be seen from Figs. \ref{fig:att} and \ref{fig:att_w}. This effect is especially prominent in pitch $(\theta)$ dynamics. 
Additionally, the optimal control action (refer to Fig. \ref{fig:ctr}) generated via the variable gain gradient descent is much smoother compared to the control action (refer to \ref{fig:ctr_w}) generated via the constant learning based method, which is found to lead to persistent chattering in control command. The control effort in both these cases is bounded within $\pm 90$ degrees.

Based on the above discussion it can be inferred that the variable gain gradient descent-based update law leads to faster convergence of critic weights closer to their ideal values. This in turn leads to achieving the ideal optimal tracking controller faster resulting into smaller overshoot in state error. Further, the control action generated by variable gain update law is devoid of chattering for the same set of actuator constraints. 

\section{Conclusion}\label{conclusion}
The paper presents a variable gain gradient descent based update law for robust optimal tracking for continuous time nonlinear systems using reinforcement learning. 
The critic neural network (NN) is utilized to approximate the value function which is also the solution of the tracking HJB equation.
It is this critic NN that is tuned online using the update law presented in this paper.
The hallmarks of this update law stems from the fact that it can adjust its learning rate based on the HJB error. 
The tuning law speeds up the learning process if the HJB error is large and it slows it down as the HJB error becomes small. 
In addition to this, the parameter update law presented in this paper leads to smaller convergence times of critic NN weights and tighter residual set over which the augmented system trajectories converge to.
The update law presented in this paper forms the basis of future scope of research using which model-free online update law to solve optimal tracking problem will be developed.


%

\section{Appendices}\label{L1}
\begin{lemma}\label{cu1}
Following equality holds true,
\begin{equation}
\begin{split}
 &2u_m\int_0^{-u_m\tanh{A(z)}}\tanh^{-1}(\nu/u_m)^TRd\nu\\
&=2u_m^2A^T(z)R\tanh{A(z)}+
u_m^2\sum_{i=1}^{m}R_i\ln[1-\tanh^2{A_i(z)}]   
\end{split}
\end{equation}
\begin{proof}
\begin{equation}
\begin{split}
\int \tanh^{-1}\Big(\frac{x}{a}\Big)=\frac{1}{2}a\ln{(a^2-x^2)}+x\tanh^{-1}\Big(\frac{x}{a}\Big)+I
\end{split}
\end{equation}
Therefore, 
\begin{equation}
\small
\begin{split}
&\int_0^u \tanh^{-1} \Big(\frac{\nu}{u_m}\Big)d\nu=\frac{1}{2}u_m\ln{(u_m^2-\nu^2)}+\nu\tanh^{-1}\Big(\frac{\nu}{u_m}\Big)\Big|_0^u\\
&2u_m\int_0^u \tanh^{-1} \Big(\frac{\nu}{u_m}\Big)d\nu=u_m^2\ln{(u_m^2-\nu^2)}+2u_m\nu\tanh^{-1}\Big(\frac{\nu}{u_m}\Big)\Big|_0^u\\
&=u_m^2\ln{(1-\frac{u^2}{u_m^2})}+2u_m^2\tanh{A(z)}\\
&=u_m^2\ln{(1-\tanh^2{A(z)})}+2u_m^2\tanh{A(z)}\\
\end{split}
\end{equation}
where, $u=-u_m\tanh{A(z)}$ is scalar. Now if $u$ is a vector, then,
\begin{equation}
\small
\begin{split}
& 2u_m\int_0^u \tanh^{-1} \Big(\frac{\nu}{u_m}\Big)Rd\nu\\
 &=2u_m^2A^T(z)R\tanh{A(z)}+
u_m^2\sum_{i=1}^{m}R_i\ln[1-\tanh^2{A_i(z)}]
\end{split}
\end{equation}
\end{proof}
\end{lemma}

\begin{lemma}\label{cu}
Following inequality holds true:
\begin{equation}
 C(u_i)=2u_m\int_{0}^{u_i}\psi^{-1}(\frac{\nu}{u_m})R_id\nu   \geq 0
\end{equation}
if $\psi^{-1}$ is monotonic odd and increasing and $R_i>0$. Where $u_i \in \mathbb{R},i=1,2,...,m$
\end{lemma}
\textbf{Proof}: If $\psi^{-1}$ is monotonic odd and increasing, then,
\begin{equation}
\begin{split}
 \Big(\frac{\nu}{u_m}\Big)\psi^{-1}\Big(\frac{\nu}{u_m}\Big) \geq 0
\end{split}
\end{equation}
or 
\begin{equation}
\nu\psi^{-1}\Big(\frac{\nu}{u_m}\Big) \geq 0
\end{equation}
where $\nu\in \mathbb{R}$ and $u_m>0$. Let $\theta=1/u_m$. 
In order to prove that, $2u_m\int_{0}^{u_i}\psi^{-1}(\nu/u_m)R_id\nu \geq 0$, it is enough to prove that, $\int_{0}^{u_i}\psi^{-1}(\nu\theta)d\nu   \geq 0$. 
In order to prove this inequality, a variable, $\mathcal{K}\in[0,\theta]$ is assumed.
Therefore,
\begin{equation}
\begin{split}
\int_{0}^{u_i}\psi^{-1}(\nu\theta)d\nu=\frac{1}{\theta}\int_{0}^{u_i\theta}\psi^{-1}(l)dl
\end{split}
\end{equation}
where $l=\nu\theta$. Similarly, 
\begin{equation}
\begin{split}
 \frac{1}{\theta}\int_{0}^{u_i\theta}\psi^{-1}(l)dl=\frac{1}{\theta}\int_{0}^{\theta}\psi^{-1}(u_i\mathcal{K})u_id\mathcal{K}
\end{split}
\end{equation}
by utilizing $l=u_i\mathcal{K}$\\
Since, $\psi^{-1}(u_i\mathcal{K})u_i \geq 0$, which implies, 
\begin{equation}
\frac{1}{\theta}\int_{0}^{\theta}\psi^{-1}(u_i\mathcal{K})u_id\mathcal{K}\geq 0
\end{equation}
\begin{lemma}\label{mean_val_lem}
Following equation holds true,
\begin{equation}
\begin{split}
u&=-u_m\tanh{\Big(\frac{1}{2u_m}R^{-1}\hat{G}^T \nabla{\vartheta}^TW+\epsilon_{uu}\Big)}=-u_m\tanh{(\tau_1(z))}\\
&+\epsilon_{u}
\end{split}
\label{ustar}
\end{equation}
where, $\small \epsilon_{uu}=(1/2u_m)R^{-1}\hat{G}^T\nabla{\varepsilon}(z)=\small [\varepsilon_{{uu}_{11}},\varepsilon_{{uu}_{12}},...,\varepsilon_{{uu}_{1m}}]^T \in \mathbb{R}^m$.
$\tau_{1}(z)=(1/2u_m)R^{-1}\hat{G}^T \nabla{\vartheta}^TW=[\tau_{11},...,\tau_{1m}]^T \in \mathbb{R}^m$ and 
$\epsilon_{u}=-(1/2)((I_m-diag(\tanh^2{(q)}))R^{-1}\hat{G}^T\nabla{\epsilon})$ 
with $q \in \mathbb{R}^m$ and $q_i \in \mathbb{R}$ considered between $\tau_{1i}+\varepsilon_{uui}$ and $\epsilon_{uui}$ i.e., $i^{th}$ element of $\epsilon_{uu}$.
\end{lemma}
\begin{proof}
\begin{equation}
u=-u_m\tanh{(\tau_1+\varepsilon_{uu})}
\end{equation}
Using mean value theorem,
\begin{equation}
\begin{split}
 &\tanh{(\tau_1+\varepsilon_{uu})}-\tanh{(\tau_1)}=\tanh^{'}{(q)}\varepsilon_{uu}\\
&=(I_m-diag(\tanh^2{(q)}))\varepsilon_{uu}   
\end{split}
\label{mmt}
\end{equation}
where, $q \in \mathbb{R}^m$ and $q_i \in \mathbb{R}$ lying between $\tau_{1i}$ and $\tau_{1i}+\varepsilon_{uui}$. 

Now, using the expression for $\varepsilon_{uu}$ in (\ref{mmt}), $\tanh{(\tau_1+\varepsilon_{uu})}$ can be rewritten as,
\begin{equation}
\begin{split}
 &\tanh{(\tau_1+\varepsilon_{uu})}=\tanh{(\tau_1)}+(I_m-diag(\tanh^2{(q)}))\\
 &\times \left(\frac{1}{2u_m}R^{-1}\hat{G}^T\nabla{\varepsilon}(z)\right)   
\end{split}
\end{equation}
Multiplying $-u_m$ on both sides,
\begin{equation}
\begin{split}
&-u_m\tanh{(\tau_1+\varepsilon_{uu})}=-u_m\tanh{(\tau_1)}-\frac{1}{2}(I_m-diag(\tanh^2{(q)}))\\
&\times \left(R^{-1}\hat{G}^T(z)\nabla{\varepsilon}(z)\right)
\end{split}
\end{equation}
Hence proved.
\end{proof}

\begin{lemma}\label{tanhlem}
Following vector inequality holds true:\\
\begin{equation}
\|\tanh{(\tau_1(z))}-\tanh{(\tau_2(z))}\| \leq T_m \leq 2\sqrt{m}
\end{equation}
where $T_m=\sqrt{\sum_{i=1}^mmin(|\tau_{1i}-\tau_{2i}|^2,4)}$, $\tau_1(z)$ and $\tau_2(z)$ both belong in $\mathbb{R}^m$, therefore, $\tanh{(\tau_i(z))} \in \mathbb{R}^m,~i=1,2$. \\
\end{lemma}
\begin{proof}
Since, $\tanh{(.)}$ is 1-Lipschitz, one can write,
\begin{equation}
\begin{split}
|\tanh{(\tau_{1i})}-\tanh{(\tau_{2i})}| \leq |\tau_{1i}-\tau_{2i}|
\end{split}
\end{equation}
Therefore using the above inequality and the fact that, $-1\leq \tanh{(.)} \leq 1$
\begin{equation}
\begin{split}
\|\tanh{(\tau_1(z))}-\tanh{(\tau_2(z))}\|^2 &= \sum_{i=1}^m|\tanh{\tau_{1i}}-\tanh{\tau_{2i}}|^2\\
& \leq \sum_{i=1}^mmin(|\tau_{1i}-\tau_{2i}|,2)^2\\
&\leq \sum_{i=1}^mmin(|\tau_{1i}-\tau_{2i}|^2,4)
\end{split}
\end{equation}
One can also see, using the absolute upper bound of $\tanh{(.)}$. 
\begin{equation}
\sum_{i=1}^mmin(|\tau_{1i}-\tau_{2i}|^2,4)\leq 2\sqrt{m}
\end{equation}
Which implies, 
\begin{equation}
\begin{split}
\|\tanh{(\tau_1(z))}-\tanh{(\tau_2(z))}\|\leq T_m \leq 2\sqrt{m} 
\end{split}
\end{equation}
\end{proof}







\bibliographystyle{iet}
\bibliography{sample}
\end{document}